\documentclass{article}

\usepackage{morewrites}

\usepackage{microtype}
\usepackage{algorithm}
\usepackage{algorithmic}
\usepackage{graphicx}
\usepackage{subcaption}
\usepackage{booktabs} 
\usepackage{multirow}

\usepackage{hyperref}
\usepackage{enumitem}

\usepackage[preprint]{neurips_2026}

\usepackage{amsmath}
\usepackage{amssymb}

\usepackage{amsthm}
\usepackage{amsfonts}
\usepackage{mathrsfs}
\usepackage{bbm}

\usepackage{mathtools}
\mathtoolsset{showonlyrefs=true}

\theoremstyle{plain}
\newtheorem{theorem}{Theorem}[section]
\newtheorem{proposition}[theorem]{Proposition}
\newtheorem{lemma}[theorem]{Lemma}

\theoremstyle{definition}

\newtheorem{assumption}[theorem]{Assumption}
\theoremstyle{remark}
\newtheorem{remark}[theorem]{Remark}
\newtheorem{example}[theorem]{Example}

\usepackage[textsize=tiny]{todonotes}

\DeclareMathOperator*{\argmax}{argmax}
\DeclareMathOperator*{\argmin}{argmin}

\newcommand{\E}{\mathbb{E}}

\newcommand{\Var}{\mathrm{Var}}

\def\T{{\textsc{t}}}

\newcommand{\cov}{\textup{cov}}
\newcommand{\var}{\textup{var}}
\newcommand{\avar}{\textup{avar}}

\newcommand{\bY}{\boldsymbol{Y}}

\newcommand{\ba}{\boldsymbol{a}}
\newcommand{\bA}{\boldsymbol{A}}
\newcommand{\bG}{\boldsymbol{G}}

\newcommand{\bX}{\boldsymbol{X}}

\newcommand{\bF}{\boldsymbol{F}}

\newcommand{\bI}{\boldsymbol{I}}

\newcommand{\bLambda}{\boldsymbol{\Lambda}}
\newcommand{\bSigma}{\boldsymbol{\Sigma}}
\newcommand{\bPhi}{\boldsymbol{\Phi}}
\newcommand{\bDelta}{\boldsymbol{\Delta}}
\newcommand{\bPi}{\boldsymbol{\Pi}}
\newcommand{\cN}{\mathcal{N}}

\newcommand{\cD}{\mathcal{D}}

\newcommand{\bepsilon}{\boldsymbol{\varepsilon}}
\newcommand{\indicator}{\mathbbm{1}}

\newcommand{\cG}{\mathcal{G}}

\newcommand{\sumi}{\sum_{i=1}^n}
\newcommand{\sumj}{\sum_{j=1}^n}

\newcommand{\sumij}{\sum_{i=1}^{n} \sum_{j=1}^{n}}
\newcommand{\sumrs}{\sum_{r=1}^{n} \sum_{s=1}^{n}}

\newcommand{\hort}{\textup{ht}}
\newcommand{\haj}{\textup{haj}}
\newcommand{\bOne}{\boldsymbol{1}}
\newcommand{\bZero}{\boldsymbol{0}}

\newcommand{\bFaug}{\bF^{(d_1\oplus d_2)}}
\newcommand{\indep}{{~\perp\!\!\!\perp~}}
\newcommand{\nonindep}{{~\not\!\perp\!\!\!\perp~}}

\usepackage[utf8]{inputenc} 
\usepackage[T1]{fontenc}    
\usepackage{hyperref}       
\usepackage{url}            
\usepackage{booktabs}       
\usepackage{amsfonts}       
\usepackage{nicefrac}       
\usepackage{microtype}      
\usepackage{xcolor}         

\title{GAUGER: Generalized Regression Adjustment via Graph-Weighted Exposure-Level Residualization for Design-Based Inference Under Interference}

\author{%
  Lei Shi\thanks{Corresponding author.} \\
  Adobe Research \\
  San Jose, CA, USA \\
  \texttt{leis@adobe.com} \\
  \And
  Rita Qiuran Lyu \\
  UC Berkeley \\
  Berkeley, CA, USA \\
  \And
  Sizhu Lu \\
  UC Berkeley \\
  Berkeley, CA, USA \\
}

\begin{document}

\maketitle

\begin{abstract}
 Estimating causal effects under interference is a common problem in social science and economics. However, it is challenging due to the complex dependency structure induced by network connections. In this paper, we propose GAUGER, a Generalized regression Adjustment framework via Graph-weighted Exposure-level Residualization for design-based causal inference under general interference. We first reveal a surprising mismatch between accuracy and efficiency in this setting: model adjustments that minimize prediction error (e.g., MSE) do not necessarily lead to the most variance reduction of the treatment effect estimator. To address this mismatch, we propose a two-step approach: (1) leveraging a strong prediction model to learn outcome patterns from the network and covariates, and (2) applying a novel calibration scheme called Graph-weighted Exposure-level Residualization (GER) that directly targets variance reduction. The resulting estimator is consistent for target causal parameters, enjoys provable variance reduction, and is asymptotically normal with a conservative variance estimator for valid statistical inference. As a practical implementation of the pipeline, we present a scheme that leverages Graph Neural Networks (GNNs) to construct the prediction model and use GER to steer the model adjustment for better variance reduction. Numerical studies show substantial efficiency gains over existing methods.
\end{abstract}

\section{Introduction}

\subsection{Background.}
Many real-world problems involve units connected through networks, such as social networks \citep{freeman2004development}, communication systems \citep{shaw1964communication}, and geographical structures \citep{rolnick2019randomized}, where outcomes may depend not only on own treatment but also on neighbors’ treatments, a phenomenon known as \emph{interference}. While interference complicates causal inference, it is essential for studying peer and spillover effects.

Randomized experiments are widely regarded as the gold standard for causal effect estimation \citep{hinkelmann2007design, imbens2015causal}. Under the design-based framework, the validity of estimation and inference is guaranteed by the known randomization mechanism with minimal assumptions on potential outcomes \citep{neyman1923application, rubin2005causal, ding2024first}. In this perspective, outcome models are used only to improve \emph{statistical efficiency}, typically via variance reduction, rather than for identification, a practice often termed as \emph{model-assisted inference}.

Recent work has extended design-based methods to interference settings \citep{hudgens2008toward, aronow2017estimating, basse2018analyzing, chang2023design, leung2022causal, gao2023causal}. However, compared to the no-interference setting, there is limited guidance on how predictive models should be used for statistically efficient adjustment under network dependence. In particular, clear guidance on efficient model-assisted adjustment under network dependence and scalable, end-to-end implementations lags behind what is available in simpler settings.

\subsection{Our contributions.}
We identify a fundamental challenge in model-assisted inference under interference: \emph{predictive accuracy does not, in general, align with statistical efficiency}. We show that estimator variance is governed by a \emph{design-induced, graph-weighted structure}, under which residuals contribute unequally depending on the network and assignment mechanism. As a result, minimizing standard prediction loss does not guarantee variance reduction.

To address this issue, we propose \emph{Generalized Regression Adjustment via Graph-weighted Exposure-level Residualization} (GAUGER), a general calibration framework that separates prediction from efficiency optimization. GAUGER first learns flexible outcome models, and then applies a calibration step that explicitly targets the graph-weighted variance structure of the estimator. We show that the calibrated estimator is provably no worse in asymptotic variance than prediction-based adjustments.

We further demonstrate how GAUGER can be combined with modern predictive models such as graph neural networks (GNNs). While GNNs provide expressive representations, we show that prediction alone is insufficient, and that graph-weighted calibration is essential for achieving efficiency. Empirical results on synthetic, semi-synthetic, and real-world network data confirm that GAUGER consistently yields substantial variance reduction over both prediction-based methods and existing graph-based approaches.

\subsection{Related work}
Prior work on causal inference under interference has established identification and estimation procedures using exposure mappings and design-based methods \citep{hudgens2008toward, aronow2017estimating, leung2022causal, gao2023causal}. However, the role of predictive modeling for \emph{statistically efficient} estimation remains underexplored.

Recent work applies GNNs to causal inference under interference, primarily in observational settings where modeling assumptions are required for identification \citep{jiang2022estimating, guo2020learning, cai2023generalization}. In contrast, we focus on randomized experiments, where identification is guaranteed and efficiency becomes the central challenge. Our framework provides a principled calibration procedure that can be applied on top of flexible predictive models, rather than relying on prediction objectives alone. See Appendix~\ref{app:related-work} for more detailed related work discussion.

\section{A Design-based setup}

We consider a finite population of $n$ units connected through a network represented by an adjacency matrix $\bG \in \{0,1\}^{n \times n}$, where $\bG_{ij}=1$ indicates a connection between units $i$ and $j$. We assume $\bG_{ii}=1$ for all $i$. Let $H_i = \sum_{j=1}^n \bG_{ij}$ denote the (self-inclusive) degree of unit $i$.  Each unit $i$ receives a binary treatment $A_i \in \{0,1\}$ with $A_i=1$ if unit $i$ is treated. Under interference, a unit’s outcome may depend on the treatments of other units. We adopt the standard \emph{exposure mapping} framework \citep{aronow2017estimating}, where this dependence is summarized through a scalar exposure variable $D_i$. Formally, for a treatment vector $\ba$ of other units, the exposure of unit $i$ is given by $D_i^{(\ba)}$, and the potential outcome depends only on the exposure level. 
\begin{assumption}[Exposure mapping]
\label{assump::exposure}
For any $i=1,\ldots,n$, $Y_i^{(\ba)} = Y_i^{(\ba')}$ whenever $D_i^{(\ba)} = D_i^{(\ba')}$.
\end{assumption}
We denote the potential outcomes as $Y_i^{(d)}$ for $d \in \cD$, where $\cD$ is the set of possible exposure levels. A commonly example is the number of treated neighbors, $D_i = \sum_{j=1,j\neq i}^n \bG_{ij} A_j.$
We denote $\bY^{(d)} = (Y_1^{(d)}, \ldots, Y_n^{(d)})^\top$ as the vector of potential outcomes under exposure level $d$.

Given the treatment assignment mechanism and exposure mapping, define 
\begin{align}
\pi_i^{(d)} &= \mathbb{P}(D_i = d), \label{eq::prob-d}\\
\pi_{ij}^{(d,d')} &= \mathbb{P}(D_i = d, D_j = d'). \label{eq::prob-dd}
\end{align}
 Intuitively, they characterize how likely each unit (and pair of units) is to fall into specific exposure level, and therefore govern how different units contribute to estimation uncertainty.As we show later, these probabilities induce a \emph{graph-dependent weighting structure} that plays a central role in determining the variance of design-based estimators. We provide further discussion on alternative exposure mappings, robustness to misspecification, and learning exposure mappings from data in Appendix~\ref{sec:exposure_mapping} abd  
how to approximate these probabilities via Monte Carlo methods \citep{aronow2017estimating}; details are provided in Appendix~\ref{sec:mc_lambda}.

Our goal is to estimate the average treatment effect between two exposure levels:
\[
\tau^{(d_1,d_2)} = \frac{1}{n} \sum_{i=1}^n \bigl( Y_i^{(d_1)} - Y_i^{(d_2)} \bigr).
\]



\section{Model adjustment under interference: a mismatch between accuracy and statistical efficiency}
\label{sec::model_assisted_adjustment} 

In design-based inference, unbiased estimators for treatment effects are readily available due to the known randomization mechanism. The primary challenge is therefore not identification, but statistical efficiency. A common approach is \emph{model-assisted adjustment}, where predictive models are used to reduce variance, often by minimizing prediction error (e.g., mean squared error).

However, under interference, this intuition breaks down.

The key issue is that estimator variance is governed by a design-induced, graph-dependent weighting structure. As a result, residuals contribute unequally depending on the network and assignment mechanism. Consequently, minimizing prediction error does not, in general, minimize estimator variance.

\paragraph{AIPW estimators under interference.}
We consider augmented inverse probability weighting (AIPW) estimators of the form
\[
\hat\tau^{(d_1,d_2)} = \hat{\mu}^{(d_1)} - \hat{\mu}^{(d_2)},
\]
where
\begin{align}\label{eq::aipw_def}
    \hat{\mu}^{(d_k)} = \frac{1}{n}\sum_{i=1}^n \left\{ \frac{\indicator(D_i=d_k)(Y_i-f_i^{(d_k)})}{\pi_i^{(d_k)}} + f_i^{(d_k)} \right\}, \quad k=1,2.
\end{align}
Here $f_i^{(d_k)}$ is any predicted potential outcome, which may depend on covariates and network structure.

A common specification is a linear model
\[
f_i^{(d_k)} = {F_i^{(d_k)}}^\top \beta_k,
\]
where $F_i^{(d_k)}$ are features (e.g., covariates or learned representations). We denote the resulting estimator by $\hat\tau^{(d_1,d_2)}_{F,\beta}$.

\paragraph{Variance structure.}
Let $\hat{1}_i^{(d_k)} = \indicator(D_i=d_k)/\pi_i^{(d_k)}$ and define $\bLambda$ as the covariance matrix of the signed vector $(\hat{\bOne}^{(d_1)}, -\hat{\bOne}^{(d_2)})$. The matrix $\bLambda$ is determined by the exposure probabilities and encodes the graph-dependent weighting structure induced by the design (see Appendix for details).

Stacking outcomes and features as
\[
\bY^{(d_1;d_2)} = 
\begin{pmatrix}
\bY^{(d_1)} \\
\bY^{(d_2)}
\end{pmatrix}, \quad
\bFaug =
\begin{pmatrix}
\bF^{(d_1)} & 0 \\
0 & \bF^{(d_2)}
\end{pmatrix},
\]
we obtain the following result.

\begin{theorem}[Mean and variance of AIPW estimators]\label{thm:mean_variance}
The estimator $\hat\tau^{(d_1,d_2)}$ is unbiased:
\[
\E[\hat\tau^{(d_1,d_2)}] = \tau^{(d_1,d_2)}.
\]
Its variance is given by
\[
n^2 \Var(\hat\tau^{(d_1,d_2)}) =
(\bY^{(d_1;d_2)} - \bFaug \beta)^\top
\bLambda
(\bY^{(d_1;d_2)} - \bFaug \beta).
\]
\end{theorem}

Theorem 1 leads to three key observations:

\begin{enumerate}[label=(\roman*)]
\item \textbf{Model-agnostic validity.} The estimator is unbiased for any choice of $f_i$, reflecting the design-based nature of the problem.

\item \textbf{Residuals drive efficiency.} Smaller residuals generally reduce variance.

\item \textbf{Prediction does not imply efficiency.} Variance depends on the quadratic form
\[
(\text{residual})^\top \bLambda (\text{residual}),
\]
where $\bLambda$ induces non-uniform, graph-dependent weighting. As a result, minimizing prediction error does not, in general, minimize estimator variance.
\end{enumerate}

This establishes a fundamental mismatch: \emph{prediction-optimal adjustments are not necessarily variance-optimal under interference}. Two models with similar prediction accuracy can yield  different estimator variances depending on how their residuals align with the graph-weighted structure. Such misalignment has also been observed in settings without interference when nonlinear models are used \citep{fogarty2018regression, guo2023generalized}. Our results reveal an additional and distinct source of mismatch arising from the interference structure itself.

This observation motivates our approach: instead of optimizing prediction loss, we construct adjustments that directly target the graph-weighted variance objective.

\section{Model construction and calibration}
\label{sec::calibration}

\subsection{The GAUGER framework}
\label{sec::gauger}
Section~\ref{sec::model_assisted_adjustment} shows that the variance of AIPW estimators is governed by a graph-weighted quadratic form:
$(\text{residual})^\top \bLambda (\text{residual}),$
where $\bLambda$ encodes the network structure and assignment mechanism. This implies that variance minimization corresponds to solving a \emph{weighted least squares problem} under the metric induced by $\bLambda$.

Given features $\bF$, the variance-optimal coefficient $\beta^*$ is given by
\begin{align}\label{eq::beta-star}
    \beta^{*}  
    = (\bFaug{}^{\top}\bLambda\bFaug)^{\dagger}\bFaug{}^{\top}\bLambda\bY^{(d_1;d_2)},
\end{align}
where for a matrix $M$, the Penrose-Moore pseudoinverse is denoted  by $M^\dagger$.

However, $\beta^*$ is infeasible since it depends on unobserved potential outcomes. Our goal is therefore to construct an estimator $\hat{\beta}$ that approximates $\beta^*$ using observed data.

A key challenge is that the weighting matrix $\bLambda$ is dense and depends on joint exposure probabilities, making direct optimization difficult in practice. Instead, we propose a tractable surrogate that captures the \emph{dependency structure} underlying $\bLambda$.

\paragraph{Graph-weighted exposure residualization (GER).}
We define a dependency graph $\bDelta \in \{0,1\}^{n \times n}$ such that
\[
\{\bDelta\}_{ij} = \indicator(D_i \not\!\perp\!\!\!\perp D_j),
\]
which encodes whether exposure assignments of units $i$ and $j$ are dependent.  For common exposure mappings, $\bDelta$ can be constructed directly from the network. For example, when the exposure
mapping depends only on self and 1-hop neighbors with Bernoulli randomization, $\bDelta$ is the $2$-hop adjacency matrix:
$\bDelta = \indicator(\bI_n + \bG + \bG^2 > 0),$
where $\bI_n$ is the $n\times n$ identity matrix. 
With this definition, we construct the following estimator: 
\begin{align}\label{eq::regression}
    (\hat{\beta}_1, \hat{\beta}_2) 
    = \argmin_{\beta_1, \beta_2} ~
     \{\hat{\bepsilon}(\beta_1, \beta_2)\}^{\T} \bDelta \{\hat{\bepsilon}(\beta_1, \beta_2)\}.
\end{align}
where
$ \hat{\bepsilon}(\beta_1, \beta_2) = 
    \hat{\bY}{}^{(d_1)} - \hat{\bY}{}^{(d_2)} 
    - \{\hat{\bF}{}^{(d_1)} - \bF{}^{(d_1)}\} \beta_1 + \{\hat{\bF}{}^{(d_2)} - \bF{}^{(d_2)}\} \beta_2,$
and $\hat{\bF}{}^{(d_k)} = \bF{}^{(d_k)} \cdot \textup{diag}\{\hat 1{}_i^{(d_k)}\}_{i=1}^n.$

This objective can be viewed as a \emph{graph-weighted residual minimization}, where $\bDelta$ approximates the dependency structure of $\bLambda$ and we use exposure-level residualized features $\hat{\bF}{}^{(d_1)} - \bF{}^{(d_1)} $ and $\hat{\bF}{}^{(d_2)} - \bF{}^{(d_2)} $ as the regressors, prioritizing reducing residuals along directions that contribute most to estimator variance. The resulting GAUGER procedure consists of two stages:

\begin{enumerate}[label=(S\arabic*)]
\item \textbf{Prediction:} Learn strong outcome models $f_i^{(d)}$ using modern machine learning methods (e.g., LASSO, random forest, GNNs, Transformer) to accurately predict the
potential outcomes for each unit at each exposure level.
\item \textbf{Calibration:} Estimate $\hat{\beta}$ via graph-weighted exposure residualization to align predictions with the variance reduction objective.
\end{enumerate}

Unlike standard approaches that rely solely on predictive accuracy, GAUGER explicitly targets the design-based efficiency criterion derived from $\bLambda$.

\subsection{Theoretical insights}
In this section, we show  that GAUGER consistently approximates the variance-optimal solution and achieves the desired efficiency guarantees.

We adopt the standard design-based asymptotic framework \citep{ding2024first, aronow2017estimating}, where we consider a sequence of finite populations 
$
\mathcal{P}_n = \{(\bY_n^{(d)})_{d\in\cD}, (\bF_n^{(d)})_{d\in\cD}, \bG_n\},
$ indexed by sample size $n$ with $n \to \infty$, and treatment is assigned via Bernoulli randomization.

We impose standard regularity conditions ensuring positivity, boundedness, and controlled dependency. In particular, we assume:
(i) exposure probabilities are bounded away from zero,
(ii) potential outcomes and features are logarithmically bounded,
(iii) the dependency graph has bounded degree, and
(iv) the graph-weighted second moments are stable after normalization.  The detailed assumption is in Appendix~\ref{sec:regularity-condition}. We first show that the GER estimator consistently recovers the variance-optimal coefficient.

\begin{theorem}[Convergence of the coefficient]\label{thm:convergence_beta}
Under regularity assumptions,
\[
\hat\beta - \beta^\star  = o_p(1).
\]
\end{theorem}

This result shows that the proposed graph-weighted residualization procedure successfully approximates the oracle solution derived from the $\bLambda$-weighted variance minimization problem.

We next establish the asymptotic distribution of the calibrated estimator.

\begin{theorem}[Asymptotic normality and efficiency]\label{thm:normality}
Under Assumption~\ref{asp::positivity_a}-\ref{asp::stability},
\[
\frac{\hat\tau^{(d_1,d_2)}_{F, \hat{\beta}} - \tau^{(d_1,d_2)}}{\sqrt{\avar\{\hat\tau^{(d_1,d_2)}_{F, \beta^*}\}}}
\;\;\xrightarrow{d}\;\; \mathcal{N}(0, 1).
\]

Moreover, the asymptotic variance satisfies
\[
\avar\{\hat\tau^{(d_1,d_2)}_{F, \beta^*}\}
\;\le\;
\avar\{\hat\tau^{(d_1,d_2)}_{F, \beta}\}
\quad \text{for any } \beta \in \mathbb{R}^p.
\]
\end{theorem}

This result provides a formal justification of GAUGER: given a feature space $\bF$, the calibrated estimator achieves the \emph{minimum asymptotic variance} among all linear adjustments of the form considered. In particular, it guarantees that the proposed calibration step is never worse, and often strictly better than standard prediction-based adjustments.

Together, these results show that GAUGER directly addresses the mismatch identified in Section~\ref{sec::model_assisted_adjustment}: while predictive models provide a useful starting point, the calibration step is essential to align the adjustment with the graph-weighted efficiency objective.
\section{GNN-assisted AIPW estimator}
\label{sec::g-aipw}

\subsection{Using GNNs within the GAUGER framework}

 Below, we introduce one architecture of GNN-assisted estimation pipeline illustrated in Figure~\ref{fig:architecture}. The pipeline consists of three main stages: feature extraction, representation learning, and outcome prediction with representation balancing.

\begin{figure}[ht!]
    \centering
    \includegraphics[width=0.6\linewidth]{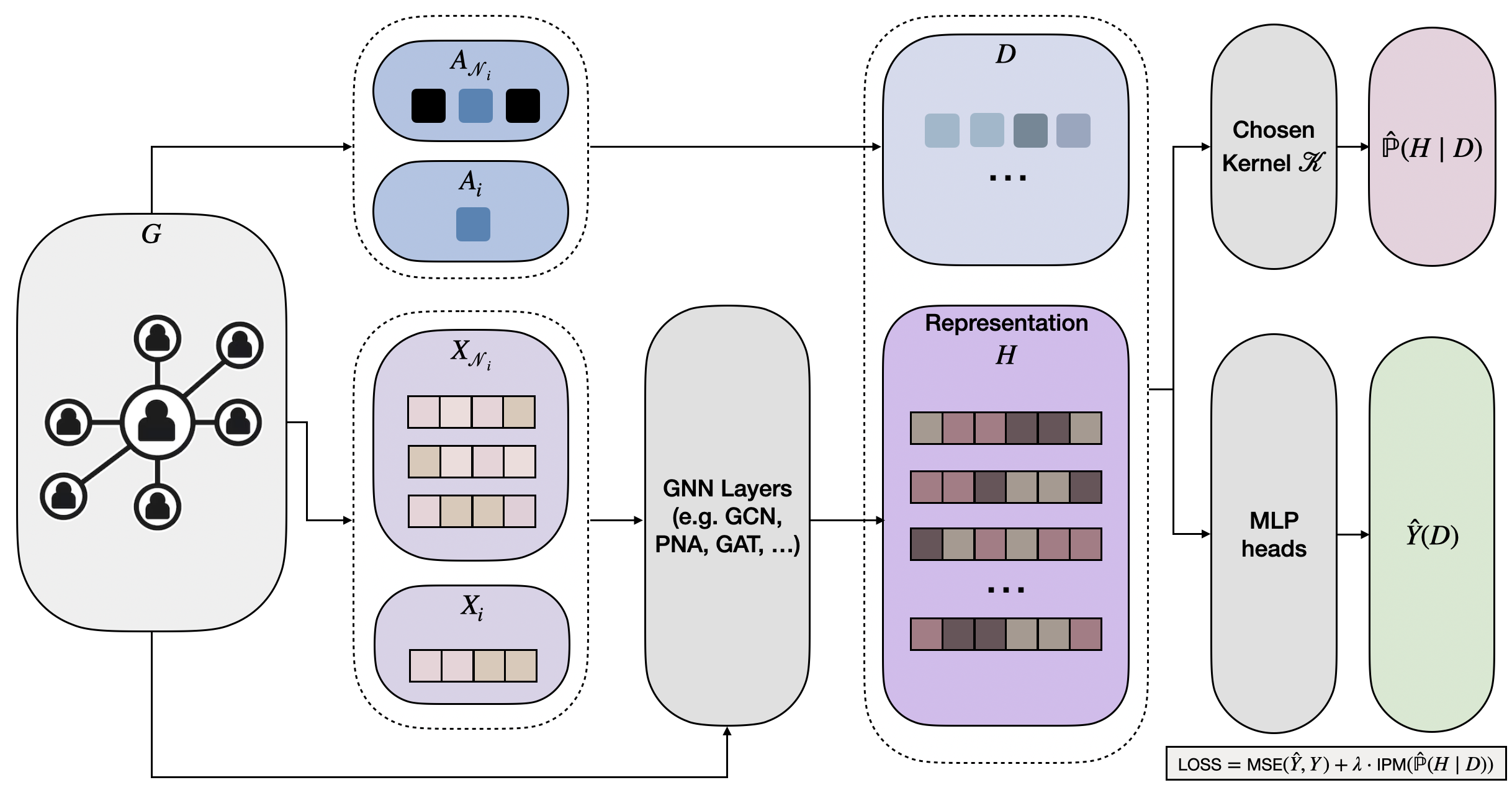}
    \caption{Overview of the GNN encoder for potential outcome prediction given exposure mapping. }
    \label{fig:architecture}
\end{figure}

\textbf{Stage 1: Feature extraction.} Given an input graph $\bG$ representing the network structure among $n$ units, we extract both individual-level and neighborhood-level information for each unit $i$. 
The treatment information is aggregated through the exposure mapping to compute the exposure level $D_i$ for each unit, which captures the relevant summary of interference from neighboring units.

\textbf{Stage 2: Representation learning via GNN.} The extracted covariates are processed through multiple layers of Graph Neural Networks (GNNs) to learn informative representations $H = (H_1, \ldots, H_n)^\T$ that encode both individual characteristics and network structure. The GNN layers perform message passing to aggregate information from neighboring nodes, allowing each unit's representation to incorporate relevant contextual information from its local network neighborhood \citep{wu2020comprehensive}. 
The framework is agnostic to the specific GNN architecture and can accommodate various choices.  
See Appendix \ref{sec:gnn_choices} for more details.

\textbf{Stage 3: Outcome prediction and representation balancing.}
The learned representations $H$ serve two purposes. First, they are fed into MLP heads to predict potential outcomes $\hat{Y}(D)$ at different exposure levels. Second, we optionally encourage representation balance across exposure groups.

Specifically, let $\widehat{P}(H \mid D=d)$ denote the empirical distribution of representations among units with exposure level $d$. We penalize discrepancies between exposure-specific distributions using an Integral Probability Metric (IPM):
\[
\mathrm{IPM}\bigl(\widehat{P}(H \mid D=d_1), \widehat{P}(H \mid D=d_2)\bigr),
\]
where $(d_1,d_2)$ are exposure levels of interest. The model is trained end-to-end with the objective
\[
\mathcal{L} = \mathrm{MSE}(\hat{Y}, Y) 
+ \lambda \cdot 
\mathrm{IPM}\bigl(\widehat{P}(H \mid D=d_1), \widehat{P}(H \mid D=d_2)\bigr),
\]
where the first term captures prediction error, and the second encourages representations to have similar distributions across exposure groups. This regularization improves stability of prediction, but does not directly target the design-based variance objective.

\subsection{GNN-assisted AIPW estimator}
The above procedure leads to prediction for each unit at each exposure level: $\hat{f}_i^{(d_k)} = \{H_i^{(d_k)}\}{}^\T \hat{\omega}_k,$ $k=1,2,$ 
where $H_i^{(d_k)}$ is the representation for unit $i$ at exposure level $d_k$, which is extracted from the second-to-last layer of the $k$-th MLP head. $\hat{\omega}_k$ is the weight for the second-to-last layer of the $k$-th MLP head. Therefore, this corresponds to taking $F_i^{(d_k)} = H_i^{(d_k)}{}^\T$ and $\hat\beta_k = \hat{\omega}_k$.

\subsection{Calibration via GAUGER}

As shown in Section~\ref{sec::model_assisted_adjustment}, prediction-based estimators are not guaranteed to be variance-optimal. We therefore apply the GER calibration procedure to align the GNN predictions with the graph-weighted variance objective. We consider two calibration paradigms.

\paragraph{Paradigm 1: Calibration of predictions.}
We take the features as
$
F_i^{(d_k)} = (1, \hat{Y}_i^{(d_k)})^\top, \quad k=1,2,$
and apply GER to estimate $\beta^*$. This corresponds to an affine transformation of the predicted outcomes, adjusting for bias in a way that targets variance reduction.

\paragraph{Paradigm 2: Calibration of representations.}
We instead take
$
F_i^{(d_k)} = (1, H_i^{(d_k)\top})^\top,$
and estimate $\beta^*$ via GER. This corresponds to modifying the final-layer weights of the GNN, allowing more flexible alignment with the variance objective.

The two paradigms represent a trade-off: calibrating predictions is more stable but less flexible, while calibrating representations provides greater flexibility at the cost of higher estimation variance.

\section{Numerical studies}
\label{sec:simulation}

We evaluate the proposed GAUGER framework across three complementary settings: (i) controlled synthetic simulations, (ii) semi-synthetic experiments on real-world networks, and (iii) a real randomized field experiment. Our goal is to assess statistical efficiency, robustness to interference, and practical applicability.

\paragraph{Synthetic study.}
We first consider a controlled simulation with network spillover effects. The data-generating process (DGP) incorporates both individual covariates and neighbor interactions, allowing us to isolate the impact of prediction and calibration. The detailed DGP is in Appendix~\ref{sec:dgp_simulation}.

We compare primitive estimators (Horvitz--Thompson and H\'{a}jek), linear adjustments, GNN-based models (GAT and PNA), and three representative GNN-based causal inference baselines: NetEstimator \citep{jiang2022estimating}, GCN-Deconfounder \citep{guo2020learning}, and RRNet \citep{cai2023generalization}.

The results are in Table~\ref{tab:synthetic_full}. We can see three consistent patterns. First, Horvitz--Thompson suffers from large variance due to unstable weights, while H\'{a}jek provides a stable baseline. Second, prediction-based adjustment alone does not guarantee improved efficiency: models trained to minimize MSE achieve similar variance to H\'{a}jek despite using additional information. Third, GAUGER calibration consistently reduces variance across all model classes. In particular, PNA-based calibrated estimators achieve the lowest RMSE.
\begin{table*}[htbp]
\centering
\caption{Synthetic study under network spillover (true $\tau^\star = -0.594$). Results are reported as mean $\pm$ 90\% Monte Carlo confidence intervals across 100 simulations.}
\footnotesize
\begin{tabular}{ll|ccc|cc}
\toprule
Class & Estimator & Bias & SD & RMSE & Coverage & Power \\
\midrule
\multirow{2}{*}{Primitive}
& $\hat{\tau}^{\text{ht}}$ 
& $-0.049 \pm 0.09$ & $0.541 \pm 0.06$ & $0.543 \pm 0.06$ & 0.60 & 0.66 \\
& $\hat{\tau}^{\text{haj}}$ 
& $0.018 \pm 0.03$ & $0.203 \pm 0.02$ & $0.203 \pm 0.02$ & 0.97 & 0.74 \\
\midrule
\multirow{3}{*}{Linear}
& $\hat{\tau}^{\text{linear-raw}}$ 
& $-0.017 \pm 0.03$ & $0.202 \pm 0.02$ & $0.203 \pm 0.02$ & 0.99 & 0.82 \\
& $\hat{\tau}^{\text{linear-linear}}$ 
& $0.034 \pm 0.03$ & $0.170 \pm 0.02$ & $0.173 \pm 0.02$ & 0.97 & 0.82 \\
& $\hat{\tau}^{\text{linear-pred}}$ 
& $0.040 \pm 0.03$ & $0.172 \pm 0.02$ & $0.176 \pm 0.02$ & 0.97 & 0.81 \\
\midrule
\multirow{3}{*}{GAT}
& $\hat{\tau}^{\text{GAT-raw}}$ 
& $0.035 \pm 0.03$ & $0.168 \pm 0.02$ & $0.172 \pm 0.02$ & 0.96 & 0.89 \\
& $\hat{\tau}^{\text{GAT-linear}}$ 
& $0.058 \pm 0.02$ & $0.141 \pm 0.02$ & $0.152 \pm 0.02$ & 0.98 & 0.95 \\
& $\hat{\tau}^{\text{GAT-pred}}$ 
& $0.058 \pm 0.02$ & $0.140 \pm 0.02$ & $0.152 \pm 0.02$ & 0.98 & 0.95 \\
\midrule
\multirow{3}{*}{PNA}
& $\hat{\tau}^{\text{PNA-raw}}$ 
& $\mathbf{-0.001 \pm 0.02}$ & $\mathbf{0.127 \pm 0.01}$ & $\mathbf{0.127 \pm 0.01}$ & 0.97 & 1.00 \\
& $\hat{\tau}^{\text{PNA-linear}}$ 
& $\mathbf{0.008 \pm 0.02}$ & $\mathbf{0.125 \pm 0.01}$ & $\mathbf{0.125 \pm 0.01}$ & 0.98 & 1.00 \\
& $\hat{\tau}^{\text{PNA-pred}}$ 
& $\mathbf{0.009 \pm 0.02}$ & $\mathbf{0.124 \pm 0.01}$ & $\mathbf{0.124 \pm 0.01}$ & 0.98 & 1.00 \\
\midrule
\multirow{3}{*}{Pure GNN Baselines}
& NetEstimator 
& $-0.090 \pm 0.02$ & $0.136 \pm 0.02$ & $0.163 \pm 0.02$ & --- & --- \\
& GCN-Deconf 
& $0.351 \pm 0.01$ & $0.088 \pm 0.01$ & $0.362 \pm 0.01$ & --- & --- \\
& RRNet 
& $-0.069 \pm 0.02$ & $0.134 \pm 0.02$ & $0.150 \pm 0.02$ & --- & --- \\
\bottomrule
\end{tabular}
\label{tab:synthetic_full}
\end{table*}
Compared to existing GNN baselines, GAUGER achieves substantially lower RMSE. GCN-Deconfounder exhibits large bias, reflecting its inability to handle spillover effects, while NetEstimator and RRNet remain suboptimal in variance. Importantly, these baseline methods do not provide valid uncertainty quantification, whereas GAUGER supports design-based inference.

\paragraph{Semi-synthetic study on real networks}

To evaluate performance in a more realistic setting, we construct a semi-synthetic experiment based on a real student friendship network from the ICPSR 37070 dataset \citep{aronow2017estimating}, consisting of $n = 2{,}983$ nodes with an average degree of approximately 2.5. 

We define exposure levels based on the fraction of treated neighbors, discretized into $k = 3$ categories $d \in \{0,1,2\}$ corresponding to low, medium, and high neighborhood treatment intensity. Individual treatments are assigned independently according to $\mathrm{Bernoulli}(1/3)$.

Potential outcomes are generated as
\begin{align}
  Y_i(d) \;&=\;
    \;\alpha
     + \beta_1 \cdot \tfrac{d}{2}
     + \beta_2 \cdot \bigl(\tfrac{d}{2}\bigr)^2 \notag 
    + \gamma_1 \cdot w_i
     + \gamma_2 \cdot \tilde{r}_i
     + \gamma_3 \cdot \tilde{r}_i^2 \notag \\
    + \delta \cdot \tfrac{d}{2} \cdot w_i & + \gamma_4 \cdot \bar{w}_i^{(1)}
     + \gamma_5 \cdot \bar{r}_i^{(1)}
     + \delta_2 \cdot \tfrac{d}{2} \cdot \bar{w}_i^{(1)} + \gamma_6 \cdot \bar{w}_i^{(2)}
     + \gamma_7 \cdot \overline{\mathrm{gpa}}_i^{(2)}
     + \delta_3 \cdot \tfrac{d}{2} \cdot \bar{w}_i^{(2)}
     + \varepsilon_i,
\end{align}
where $\varepsilon_i \sim \mathcal{N}(0, \sigma^2)$, $w_i$ denotes baseline wrist measurement, $\tilde{r}_i$ is normalized degree, and $\bar{w}_i^{(k)}$, $\bar{r}_i^{(k)}$ represent $k$-hop neighborhood averages of covariates. The parameter values are set as in with $\sigma = 0.20$.

The target estimand is the peer effect
\[
\tau^* = \mathbb{E}[Y_i(d{=}2)] - \mathbb{E}[Y_i(d{=}0)] \approx 0.4623,
\]
which captures the effect of moving from zero to high neighborhood exposure.

Across $B=200$ randomizations, we observe that all GER-based estimators achieve negligible bias, confirming their design-based validity. Compared to the linear GER baseline, GNN-enhanced GER methods (GAT and PNA) achieve substantial variance reduction, leading to approximately $2$--$2.5\times$ improvements in RMSE. Among these, PNA-based variants consistently attain the lowest RMSE, highlighting the benefit of expressive representations when combined with variance-aware calibration.

We further compare against three representative GNN-based causal inference methods: NetEstimator \citep{jiang2022estimating}, GCN-Deconfounder \citep{guo2020learning}, and RRNet \citep{cai2023generalization}. These methods are adapted to our exposure-mapping setting and evaluated over $B=200$ randomizations.

The results are in Table~\ref{tab:semi_synthetic_full}. We can see that GAUGER-based estimators significantly outperform these baselines in terms of RMSE, achieving approximately $3.5$--$4\times$ improvement over the best competing method (RRNet). In contrast, GCN-Deconfounder exhibits severe bias, as it does not explicitly model spillover effects. NetEstimator and RRNet perform more competitively but remain suboptimal.

A key reason for this gap is that these methods do not incorporate design-based weighting, and thus rely heavily on the correctness of the outcome model. In contrast, GAUGER explicitly aligns adjustment with the design-induced variance structure, leading to both low bias and low variance.

These results demonstrate that, even in realistic network settings, improving predictive representations alone is insufficient for efficient estimation under interference. Explicit calibration toward the design-based variance objective is essential for achieving optimal performance.
\begin{table*}[htbp]
\centering
\caption{Semi-synthetic study on real network ($n=2{,}983$, $\tau^* = 0.4623$). 
Results are reported as mean $\pm$ 95\% Monte Carlo confidence intervals.}
\small
\begin{tabular}{llccc}
\toprule
Category & Method & Bias & SD & RMSE \\
\midrule
\multirow{3}{*}{Primitive / Linear}
& HT 
& $-0.032 \pm 0.036$ 
& $0.091 \pm 0.026$ 
& $0.096 \pm 0.027$ \\
& H\'{a}jek 
& $-0.012 \pm 0.021$ 
& $0.052 \pm 0.015$ 
& $0.054 \pm 0.015$ \\
& Linear GER 
& $+0.006 \pm 0.012$ 
& $0.030 \pm 0.008$ 
& $0.030 \pm 0.009$ \\
\midrule
\multirow{3}{*}{GAT + GER}
& GAT-linear GER 
& $+0.001 \pm 0.005$ 
& $0.013 \pm 0.004$ 
& $0.013 \pm 0.004$ \\
& GAT-pred GER 
& $+0.000 \pm 0.005$ 
& $0.012 \pm 0.003$ 
& $0.012 \pm 0.003$ \\
& GAT-raw GER 
& $-0.003 \pm 0.005$ 
& $0.013 \pm 0.004$ 
& $0.014 \pm 0.004$ \\
\midrule
\multirow{3}{*}{PNA + GER}
& PNA-linear GER 
& $-0.002 \pm 0.005$ 
& $0.012 \pm 0.003$ 
& $0.012 \pm 0.003$ \\
& PNA-pred GER 
& $\mathbf{-0.002 \pm 0.004}$ 
& $\mathbf{0.011 \pm 0.003}$ 
& $\mathbf{0.011 \pm 0.003}$ \\
& PNA-raw GER 
& $-0.004 \pm 0.004$ 
& $\mathbf{0.010 \pm 0.003}$ 
& $\mathbf{0.011 \pm 0.003}$ \\
\midrule
\multirow{3}{*}{External GNN baselines}
& NetEstimator 
& $+0.044 \pm 0.004$ 
& $0.031 \pm 0.003$ 
& $0.053 \pm 0.004$ \\
& GCN-Deconf 
& $-0.197 \pm 0.002$ 
& $0.018 \pm 0.002$ 
& $0.198 \pm 0.002$ \\
& RRNet 
& $+0.022 \pm 0.006$ 
& $0.042 \pm 0.004$ 
& $0.047 \pm 0.005$ \\
\bottomrule
\end{tabular}
\label{tab:semi_synthetic_full}
\end{table*}
\paragraph{Real application}

We apply our GAUGER framework to analyze data from a large-scale randomized controlled trial studying the effects of unconditional cash transfers in rural Kenya \citep{egger2022general}. The GiveDirectly experiment implemented a randomized saturation design across 653 villages in Siaya County, Kenya. In this design, sublocations (administrative units containing multiple villages) were first randomly assigned to either high-saturation (two-thirds of villages treated) or low-saturation (one-third of villages treated) conditions. Within each sublocation, villages were then randomly assigned to treatment or control. All eligible households in treated villages received unconditional cash transfers of approximately \$1,000 USD. We construct a spatial network among villages based on geographic proximity, following our desciption in Appendix \ref{sec:village-network}.

We define the exposure mapping as $D_i = (A_i, S_i)$, where $A_i \in \{0, 1\}$ is the village's own treatment indicator, and $S_i = \indicator({|\cN_i|^{-1}}{\sum_{j \in \cN_i} A_j} > {1}/{3})$ indicates whether more than one-third of village $i$'s neighbors are treated. This yields four exposure levels: $(A_i, S_i) \in \{(0,0), (0,1), (1,0), (1,1)\}$. The exposure $(0,0)$ corresponds to untreated villages with few treated neighbors, while $(1,1)$ corresponds to treated villages with many treated neighbors. The contrast between $(1,0)$ and $(0,0)$ captures the direct effect of treatment with low neighborhood saturation, while $(0,1)$ versus $(0,0)$ captures the spillover effect on untreated villages from having many treated neighbors.

Table~\ref{tab:revenue_results} presents treatment effect estimates for revenue per enterprise (in Kenyan Shillings or KES). The H\'{a}jek, Linear, and PNA-no-harm estimators yield consistent conclusions: a statistically significant positive direct effect when neighbors are mostly treated ($H=1$), with point estimates around 1,356--1,391 KES. In contrast, the Horvitz-Thompson estimator produces misleading results---it shows spurious significance for Direct ($H=0$) and Spillover ($A=0$) due to its sensitivity to extreme inverse probability weights, while missing the true Direct ($H=1$) effect. Table~\ref{tab:revenue_results} presents 95\% confidence interval widths for the treatment effect estimates, showing that the PNA-no-harm estimator achieves the narrowest average confidence interval width for all four contrasts. This efficiency gain comes from the PNA model's ability to leverage network structure for outcome prediction, reducing residual variance.

\begin{table*}[htbp]
    \centering
    \caption{Treatment effect estimates for revenue per enterprise (KES). * denotes statistical significance at $\alpha=0.05$.}
    \label{tab:revenue_results}
    \footnotesize
    \begin{tabular}{l cc cc cc cc}
        \toprule
        & \multicolumn{2}{c}{HT} & \multicolumn{2}{c}{H\'{a}jek} & \multicolumn{2}{c}{Linear} & \multicolumn{2}{c}{PNA-no-harm} \\
        \cmidrule(lr){2-3} \cmidrule(lr){4-5} \cmidrule(lr){6-7} \cmidrule(lr){8-9}
        Contrast & $\hat\tau$ & 95\% CI & $\hat\tau$ & 95\% CI & $\hat\tau$ & 95\% CI & $\hat\tau$ & 95\% CI \\
        \midrule
        Direct ($H=0$) & 2786* & [405, 5167] & 841 & [-1085, 2767] & 879 & [-1084, 2842] & 731 & [-1088, 2549] \\
        Direct ($H=1$) & 113 & [-1185, 1410] & 1391* & [105, 2678] & 1356* & [96, 2617] & 1390* & [186, 2594] \\
        Spillover ($A=0$) & 3770* & [1963, 5577] & 46 & [-1372, 1464] & -75 & [-1479, 1329] & -64 & [-1402, 1275] \\
        Spillover ($A=1$) & 1097 & [-774, 2969] & 596 & [-1198, 2390] & 481 & [-1508, 2470] & 536 & [-1215, 2288] \\
        \bottomrule
    \end{tabular}
\end{table*}

Across all settings, the results consistently support the central thesis of this paper: improving predictive accuracy alone is insufficient for efficient estimation under interference. Explicit calibration toward the graph-weighted variance objective is necessary to translate predictive gains into statistical efficiency.

\section{Discussion}
In this paper, we have proposed a GAUGER framework for estimating treatment effects under interference. We have shown that the GAUGER framework can be used to estimate treatment effects under interference and provide a principled way to construct model adjustment for statistical efficiency. We have also shown that the GAUGER framework can be implemented using GNNs to achieve substantial efficiency gains over existing methods.

There are many remaining questions and extensions to the current work. For example, most exposure mappings are discrete, and it is worth exploring how to extend the GAUGER framework to continuous exposure mappings. Moreover, it is an interesting question to explore how to use GAUGER for facilitating a better experimental design on the networks for better effect estimation.

\bibliography{example_paper}
\bibliographystyle{apalike}

\appendix
\section{Detailed related work}\label{app:related-work}
There is a large literature on causal inference with interference. Early work by \citet{hudgens2008toward} formalized interference in two-stage randomized designs, while \citet{aronow2017estimating} built around \emph{exposure mappings} that summarize peers’ treatments
and enable well-defined estimands and estimators. More recently, \citep{leung2022causal, gao2023causal} developed asymptotic theory under approximate neighborhood interference, and specialized settings such as bipartite experiments, clustered designs, and two-sided randomization have also been studied \citep{papadogeorgou2025causal, lu2025design, su2021model, Masoero2026multiple}.

While this literature establishes identification and valid estimation procedures under interference, much less attention has been paid to \emph{efficiency} in model-assisted settings. In particular, existing design-based approaches provide limited guidance on how to use predictive models to achieve optimal variance reduction, especially in the presence of network dependence. Our work addresses this gap by explicitly characterizing the relationship between prediction and efficiency and proposing a calibration framework that targets the design-based variance.

\textbf{Graph neural networks for causal inference with interference.}
A growing body of work applies graph neural networks (GNNs) to causal inference under interference \citep{wu2025causal, jiang2022estimating, guo2020learning, guo2021ignite, cai2023generalization, chen2024doubly, ma2021causal, du2025telling, ma2022learning, leung2022graph} in observational studies. These methods typically learn representations or outcome models using predictive or balancing objectives, where identification relies on modeling assumptions such as unconfoundedness.

In contrast, our work focuses on randomized experiments, where identification is guaranteed by the design and the primary challenge is \emph{statistical efficiency} rather than bias reduction. As a result, the role of modeling is fundamentally different: instead of learning representations for identification, we aim to calibrate predictions to minimize the variance of a design-based estimator. Our proposed GAUGER framework is therefore not tied to a specific model class, but provides a principled post-hoc calibration that can be applied on top of flexible predictive models, including GNNs. Empirically, we show that raw prediction-only GNN methods can be suboptimal from an efficiency perspective, and that our graph-weighted calibration leads to substantial variance reduction.
\section{variance structure}\label{sec:variance-structure}

\paragraph{Variance structure.}
To characterize the variance of AIPW estimator under interference, we introduce the notation $\bLambda$, which corresponds to the covariance matrix of the signed reweighted exposure vector
$(\{\hat \bOne{}^{(d_1)}\}^\T, -\{\hat \bOne{}^{(d_2)}\}^\T)^\T$, where for $k=1,2$
\begin{align*}
    \hat \bOne{}^{(d_k)} = \{\hat 1{}_i^{(d_k)}\}_{i=1}^n, \quad \hat 1{}_i^{(d_k)} = \frac{\indicator(D_i=d_k)}{\pi_i(d_k)}. 
\end{align*}

$\bLambda$ is given by the following $2n\times 2n$ block matrix
\begin{align}\label{eq::bLambda}
    \bLambda =
    \begin{pmatrix}
        \bLambda_{11} & \bLambda_{12} \\
        \bLambda_{12} & \bLambda_{22}
    \end{pmatrix}
\end{align}
with the two diagonal blocks being $n\times n$ symmetric matrices with $(i,j)$-th item equal to

\begin{align*}
    \{\bLambda_{11}\}_{(i,j)} =& \indicator(j=i) \frac{1-\pi_i^{(d_1)}}{\pi_i^{(d_1)}} 
    + \indicator(j\neq i)\frac{\pi_{ij}^{(d_1,d_1)}-\pi_i^{(d_1)}\pi_j^{(d_1)}}{\pi_i^{(d_1)}\pi_j^{(d_1)}}, \\
    \{\bLambda_{22}\}_{(i,j)} =& \indicator(j=i) \frac{1-\pi_i^{(d_2)}}{\pi_i^{(d_2)}} 
    + \indicator(j\neq i)\frac{\pi_{ij}^{(d_2,d_2)}-\pi_i^{(d_2)}\pi_j^{(d_2)}}{\pi_i^{(d_2)}\pi_j^{(d_2)}},
\end{align*}

and the off-diagonal blocks being an $n\times n$ symmetric matrix with $(i,j)$-th item equal to

$$
    \{\bLambda_{12}\}_{(i,j)} =\indicator{(j=i)} 
    - \indicator{(j\neq i)}\frac{\pi_{ij}^{(d_1,d_2)}-\pi_i^{(d_1)}\pi_j^{(d_2)}}{\pi_i^{(d_1)}\pi_j^{(d_2)}}.
$$
\section{Regulariy assumptions}\label{sec:regularity-condition}
 Define the oracle residuals
\[
\varepsilon^{(d_k)} = \bY^{(d_k)} - \bF^{(d_k)} \beta_k^*,
\]
where $\beta^*$ is the variance-optimal coefficient defined in Section~\ref{sec::gauger}.

We further define the graph-weighted population second-order moments
of the potential outcomes:
\begin{gather*}
\bSigma^{\bY}_n = \{\bY^{(d_1 \oplus d_2)}\}^{\top} \bLambda \{\bY^{(d_1 \oplus d_2)}\}, 
\bSigma^{\bF}_n = \{\bF^{(d_1 \oplus d_2)}\}^{\top} \bLambda \{\bF^{(d_1 \oplus d_2)}\}, 
\bSigma^{\varepsilon}_n = \{\varepsilon^{(d_1 \oplus d_2)}\}^{\top} \bLambda \{\varepsilon^{(d_1 \oplus d_2)}\}.
\end{gather*}
\begin{assumption}
    Let $\gamma_1, \gamma_2, \gamma_3 > 0$ be universal constants. We assume the following regularity conditions: 
    \begin{enumerate}[label=(\alph*), ref=\theassumption (\alph*)]
        \item \label{asp::positivity_a} \textit{  First-order inclusion positivity}: The propensity scores satisfy strong positivity with some  $\eta_n = O((\log n)^{\gamma_1})$: $
        \min_{d\in\cD} \min_{i\in[n]} \pi_i^{(d)} > \eta_n^{-1}$.

        \item \label{asp::positivity_b} \textit{Second-order inclusion positivity}: The minimal positive joint inclusion probability is bounded below by some $\eta_n^{-1}$: $\min_{d\neq d'\in\cD} \min_{i,j} \{\pi_{ij}^{(d;d')}: \pi_{ij}^{(d;d')} > 0\} > \eta_n^{-1}$. 

        \item \label{asp::bound} \textit{Logarithmically bounded potential outcomes and covariates}: The potential outcomes and covariates are bounded with some $M_n = O((\log n)^{\gamma_2})$: $\max_{i\in[n]} |Y_i| < M_n, ~ \max_{i\in[n], d\in\cD} \|F_i^{(d)}\|_2 < M_n$.

        \item \label{asp::bounded_degree} \textit{Logarithmically bounded depenency degree}: The network has a bounded degree with some $\delta_n = O((\log n)^{\gamma_3})$: $\max_{i\in[n]} \deg_i(\bDelta) < \delta_n$.

        \item \label{asp::stability} \textit{Stability}: 
        There is a rate function $\kappa_n$ that is at least constant order, i.e., $ \underset{n\to\infty}{\liminf} ~\kappa_n > 0$, such that the following quantities are asymptotically of constant order: for some constants $\underline{c}, \overline{c} > 0$, and some interger $r \le p$,
        \begin{gather}
            \underline{c}\bI_2\precsim\frac{1}{n\kappa_n}\bSigma^{\varepsilon}_n \precsim \overline{c}\bI_2, \quad \underline{c}\bI_2\precsim\frac{1}{n\kappa_n}\bSigma^{\bY}_n \precsim \overline{c}\bI_2, \\ \underline{c}\bI_r\precsim\frac{1}{n\kappa_n}\bSigma^{\bF}_n \precsim \overline{c}\bI_r, \text{ for some integer } 1 \le r \le p .
        \end{gather}
    \end{enumerate}
\end{assumption}
These assumptions are standard in design-based inference with interference and ensure that the influence of any individual unit remains controlled as $n$ grows.

\section{Additional details}
\subsection{Monte Carlo approximation for $\bLambda$}\label{sec:mc_lambda}
One key step for the framework is to perform Monte Carlo approximation to estimate the matrix $\bLambda$. We use the Monte Carlo simulation approach introduced by Section 3 of \citet{aronow2017estimating}. Introduce the following matrix $\bPi^{(d_k,d_l)}\in\mathbb{R}^{n\times n}$ to encode the probabailities in \eqref{eq::prob-d} and \eqref{eq::prob-dd}, with elements:
\begin{align}
    \{\bPi^{(d_k, d_l)}\}_{ij} =
    \left\{
        \begin{array}{ll}
            \pi_i^{(d_k)}, & i=j,~ k=l;\\
            \pi_{ij}^{(d_k, d_l)}, & \text{otherwise}.
        \end{array}
    \right.
\end{align}
Based on \eqref{eq::bLambda}, $\bLambda$ can be directly computed from $\bPi^{(d_k, d_l)}$. Algorithm \ref{alg:MC_Lambda} summarizes the procedure to compute $\bLambda$. 
\begin{algorithm}[ht!]
    \caption{Monte Carlo Simulation for Computing $\bLambda^{(d_k, d_l)}$}
    \label{alg:MC_Lambda} 
 \begin{algorithmic}[1]
    \STATE {\bfseries Input:} Randomization plan $\mathcal{P}$, rounds of Monte Carlo repetition $R$, number of exposure levels $K$, number of unites $n$, network $\bG$ 
    \STATE Perform $R$ rounds of randomization following $\mathcal{P}$, to construct $K$ indicator-of-exposure matrix in $\{0,1\}^{n\times R}$, given by 
    $
    \mathbbm{1}^{\textup{MC}}_k = \{1^{(r)}(D_i = d_k)\}_{i\in[n], r\in[R]}, 
    $
    where $1^{(r)}(D_i = d_k)$ indicates whether unit $i$ has exposure level $d_k (k=1,2)$ in Round $r$.

    \FOR{k=1,2} 
    \STATE Compute $\hat\bPi{}^{(d_k,d_k)} = \{\mathbbm{1}^{\textup{MC}}_k \cdot (\mathbbm{1}^{\textup{MC}}_k)^\T + I_n\}/(R+1)$. The $(i,i)$-th diagonal of $\hat\bPi{}^{(d_k,d_k)}$ approximates the propensity scores $\pi_i^{(d_k)}$, and the $(i,j)$-th off-diagonal value approximates the joint inclusion probability $\pi_{ij}^{(d_k,d_k)}$. 
    \ENDFOR
    
    \STATE Compute $\hat\bPi{}^{(d_k,d_l)} = \{\mathbbm{1}^{\textup{MC}}_k \cdot (\mathbbm{1}^{\textup{MC}}_l)^\T\}/R$. The $(i,j)$-th value approximates the joint inclusion probability $\pi_{ij}^{(d_k,d_l)}$.
     
    \STATE Use the definition \eqref{eq::bLambda} to construct $\hat\bLambda^{(d_k,d_l)}$ from $\hat\bPi^{(d_k,d_l)}$. 
    \STATE \textbf{Output: }  $\hat\bLambda{}^{(d_k,d_l)}$ and $\hat\bPi{}^{(d_k,d_l)}$.
 \end{algorithmic}
 \end{algorithm}

\citet{aronow2017estimating} also provides a theoretical justification for Algorithm \ref{alg:MC_Lambda}, which suggests that the Monte Carlo approximation is asymptotically consistent:
\begin{proposition}[Proposition 3.1 in \citet{aronow2017estimating}]
    As $R \rightarrow \infty$, $\hat{\bPi}{}^{(d_k,d_l)} \rightarrow \bPi^{(d_k,d_l)}$ almost surely. Thus, we also have that $\hat\bLambda{}^{(d_k,d_l)} \rightarrow \bLambda^{(d_k,d_l)}$ almost surely.
\end{proposition}

\begin{remark}[Computational complexity of Algorithm \ref{alg:MC_Lambda}]\label{remark:mc_lambda}
    The main computational cost of Algorithm \ref{alg:MC_Lambda} is the computation of the matrix $\hat\bPi{}^{(d_k,d_l)}$, which requires $O(n^2R)$ operations. More broadly, if we apply Algorithm \ref{alg:MC_Lambda} to compute the matrix $\bLambda$ for $K$ exposure levels, the computational complexity is $O(n^2K^2R)$.
\end{remark}

\subsection{Variance estimation}\label{sec:variance_estimation}
Following the tradition of design-based inference, it is not possible to construct a consistent variance estimator for $\hat\tau^{(d_1, d_2)}_{F, \hat\beta}$, because we can never observe and estimate the cross-level covariance between potential outcomes. Instead, we need to construct a variance estimator with a certain level of conservativeness to be robust to different possible cross-level correlations. 

First, we define the following matrix:
\begin{align*}
    \bPhi_{kl} = \textup{diag}\{\phi_{i+}^{(d_k, d_l)}\}_{i=1}^n, \quad k = 1,2, 
\end{align*}
where $\phi_{i+}^{(d_k, d_l)} = \sumj \indicator(\pi_{ij}^{(d_k, d_l)} = 0, \bDelta_{ij} = 1) $ counts the number of units $j$ under level $d_l$ that have $D_j \nonindep D_i$ and are never jointly included in the sample with $i$ under exposure $d_k$. The reweighted version is given by
\begin{align*}
    \bPhi_{kl}^\pi = \textup{diag}\{\pi_i^{(d_k)}\phi_{i+}^{(d_k, d_l)}\}_{i=1}^n,
\end{align*}

Then, we define the following matrix, which is a augmented version of $\bLambda$: 
\begin{eqnarray*}
    \{\bLambda_{kl}^\pi\}_{(i,j)}
    =  
    \left\{
        \begin{array}{cc}
            \{\bLambda_{kl}\}_{(i,j)} \cdot \frac{\pi_i^{(d_k)}\pi_i^{(d_l)}}{\pi_{ij}^{(d_k,d_l)}}, & \pi_{ij}^{(d_k,d_l)} \neq 0; \\
            0, & \pi_{ij}(d_k,d_l) = 0. 
        \end{array}
    \right.    
\end{eqnarray*}    
Meanwhile, we construct the following variance estimator:
\begin{align}\label{eq::var_est}
    \hat{v}^{(d_1,d_2)} = \left(
        \sqrt{\hat{v}^{(d_1)}} + \sqrt{\hat{v}^{(d_2)}}
    \right)^2,
\end{align}
where for $k = 1,2$,
{\small 
\begin{align*}
    &\hat{v}^{(d_k)} = n^{-2} \{\hat{\bY}{}^{(d_k)} - \hat{\bF}{}^{(d_k)}\hat\beta_k\}^\T (\bLambda_{kk}^\pi + \bPhi_{kk}^\pi) \{\hat{\bY}{}^{(d_k)} - \hat{\bF}{}^{(d_k)}\hat\beta_k\}. 
\end{align*}
}
The variance estimator motivates the following confidence interval:
\begin{align*}
    \hat\tau^{(d_1,d_2)}_{F, \hat{\beta}} \pm z_{\alpha/2} \sqrt{\hat{v}^{(d_1,d_2)}},
\end{align*}
where $z_{\alpha/2}$ is the $1-\alpha/2$ quantile of the standard normal distribution.

The following theorem shows the conservativeness of the variance estimator $\hat{v}(d_1,d_2)$.
\begin{theorem}[Conservative variance estimator and CI]\label{thm:var_est}
    Assume Assumption \ref{asp::positivity_a}-\ref{asp::stability}, and a consistent estimator $\hat\beta$ of $\beta^\star$, i.e., $\hat\beta - \beta^\star = o_p(1)$. Define 
    \begin{eqnarray*} 
        v_{\lim}^{(d_1, d_2)} = \left\{\sqrt{v_{\lim}^{(d_1)}} + \sqrt{v_{\lim}^{(d_2)}}\right\}^2, 
    \end{eqnarray*}
    where $v_{\lim}^{(d_k)}$ is given in Appendix \ref{sec:pf-var-est}.   
    Then the variance estimator \eqref{eq::var_est} is conservative:
    \begin{eqnarray*}
        \frac{\hat{v}^{(d_1,d_2)} - v_{\lim}^{(d_1,d_2)}}{v_{\lim}^{(d_1,d_2)}} = o_p(1), \quad v_{\lim}^{(d_1,d_2)} \ge \avar\{\hat\tau^{(d_1,d_2)}_{F, \hat\beta}\}.
    \end{eqnarray*}
\end{theorem}
The conservativeness of the variance estimator also guarantees the validity of the confidence interval. 

\subsection{Computation complexity for hyperparameter construction}\label{sec:computation}

In this section, we comment on the computation complexity of the hyperparameter construction, when there are $K$ pairs of contrast exposures we want to compare.

Applying GER to estimate $\beta^\star$ only requires computation of $\bDelta$ and the first-order inclusion probabilities \eqref{eq::prob-d}, and there is no need for the second-order inclusion probabilities \eqref{eq::prob-dd}. This improves the computation complexity to $O(nKR)$ for first-order inclusion probabilities where $R$ is the number of rounds of Monte Carlo simulation, plus $O(n \delta_n^2)$ for constructing $\bDelta$, resulting in a total computational complexity of $O(nKR + n \delta_n^2)$ for hyperparameter construction.

Another straightforward approach is to estimate $\beta^\star$ by a plug-in approach using IPW counterpart for $\bY^{(d_1;d_2)}$, as well as $\bFaug$. Such an approach has been discussed in multiple literature for design-based causal inference, such as \citet{middleton2018unified, chang2023design, lu2025design}. Based on our discussion in Appendix \ref{sec:mc_lambda}, when there are $K$ pairs of contrast exposures, the computational complexity is $O(K n^2 R)$, where $R$ is the number of rounds of Monte Carlo simulation. 

\subsection{Choice of GNN architecture}\label{sec:gnn_choices}

Our framework is agnostic to the specific GNN architecture. We describe common choices following the message-passing paradigm where unit $i$'s representation $H_i^{(\ell)}$ at layer $\ell$ is updated by aggregating from neighbors $\cN_i$. Let $H_i^{(0)} = X_i$.

\begin{enumerate}[leftmargin=*]
    \item \textbf{GCN} \citep{kipf2016semi}: Normalized mean aggregation. $H_i^{(\ell+1)} = \sigma\big( W^{(\ell)} \cdot \frac{1}{|\cN_i|+1} ( H_i^{(\ell)} + \sum_{j \in \cN_i} \frac{H_j^{(\ell)}}{\sqrt{|\cN_j|+1}} ) \big)$. Efficient but treats all neighbors equally.
    
    \item \textbf{GAT} \citep{velivckovic2017graph}: Learned attention weights. $H_i^{(\ell+1)} = \sigma\big( \sum_{j \in \cN_i \cup \{i\}} \alpha_{ij}^{(\ell)} W^{(\ell)} H_j^{(\ell)} \big)$, where $\alpha_{ij}$ are softmax-normalized attention scores computed from node pairs. Focuses on informative neighbors.
    
    \item \textbf{GIN} \citep{xu2018powerful}: Maximally expressive. $H_i^{(\ell+1)} = \textup{MLP}^{(\ell)}\big( (1 + \epsilon^{(\ell)}) H_i^{(\ell)} + \sum_{j \in \cN_i} H_j^{(\ell)} \big)$. Sum aggregation with MLP distinguishes more graph structures.
    
    \item \textbf{PNA} \citep{corso2020principal}: Multiple aggregators with degree scalers. $H_i^{(\ell+1)} = U^{(\ell)} \big( H_i^{(\ell)}, \bigoplus_{\textup{agg}, s} s(|\cN_i|) \cdot \textup{agg}(\{H_j^{(\ell)}\}_{j \in \cN_i}) \big)$, where $\textup{agg} \in \{\textup{mean}, \textup{max}, \textup{min}, \textup{std}\}$ and $s \in \{1, |\cN_i|^{-1/2}, |\cN_i|^{-1}\}$. Robust across different graph types.
\end{enumerate}

\subsection{A cross-fitting extension}\label{sec:cf}

Cross-fitting is a standard technique to reduce overfitting bias when the same data is used for both model fitting and estimation \citep{lu2025conditional}. Under network interference, the dependency structure requires a more careful sample splitting strategy. We propose a graph-based cross-fitting procedure tailored to interference settings, summarized in Algorithm \ref{alg:cross_fitting}.

\begin{algorithm}[ht!]
    \caption{Graph-Based Cross-Fitting for Network Interference}
    \label{alg:cross_fitting}
\begin{algorithmic}[1]
    \STATE {\bfseries Input:} Network $\bG$, covariates $\bX$, outcomes $\bY$, exposures $D_1,\ldots,D_n$, propensities $\pi_i^{(d)}$, balance tolerance $\epsilon > 0$
    
    \STATE \textbf{(Construct 2-hop graph)} Compute $\bG^{(2)}$ where $\bG^{(2)}_{ij} = 1$ if $\cN_i \cap \cN_j \neq \emptyset$ (i.e., units $i$ and $j$ share a common neighbor in $G$).
    
    \STATE \textbf{(Balanced min-cut)} Partition units into $\mathcal{V}_1, \mathcal{V}_2$ by solving:
    $
    \min_{\mathcal{V}_1, \mathcal{V}_2} |\{(i,j) : G^{(2)}_{ij} = 1, i \in \mathcal{V}_1, j \in \mathcal{V}_2\}| ~ \text{s.t.} ~ ||\mathcal{V}_1| - |\mathcal{V}_2|| \leq \epsilon n.
    $
    
    \STATE \textbf{(Create buffer zones)} Define training sets by removing boundary units:
    $\mathcal{T}_1 = \mathcal{V}_1 \setminus \{i : \exists j \in \mathcal{V}_2 \text{ with } G^{(2)}_{ij} = 1\}$, and
    $\mathcal{T}_2 = \mathcal{V}_2 \setminus \{i : \exists j \in \mathcal{V}_1 \text{ with } G^{(2)}_{ij} = 1\}$.
    
    \STATE \textbf{(Fit models)} Train outcome model $\hat{f}_1$ on units in $\mathcal{T}_1$; train $\hat{f}_2$ on units in $\mathcal{T}_2$.
    
    \STATE \textbf{(Cross-predict)} For each unit $i$ and exposure level $d$, compute:
    $
    \hat{Y}_i^{(d),\textup{cf}} = \hat{f}_2(X_i, G; d) \cdot \indicator(i \in \mathcal{V}_1) + \hat{f}_1(X_i, G; d) \cdot \indicator(i \in \mathcal{V}_2).
    $
    
    \STATE \textbf{(Construct estimator)} Compute the cross-fitted AIPW estimator:
    $
    \hat{\mu}^{(d)}_{\textup{cf}} = \frac{1}{n}\sumi \left\{ \frac{\indicator(D_i=d)(Y_i - \hat{Y}_i^{(d),\textup{cf}})}{\pi_i^{(d)}} + \hat{Y}_i^{(d),\textup{cf}} \right\}.
    $
    
    \STATE \textbf{Output:} Cross-fitted estimator $\hat{\tau}^{(d_1,d_2)}_{\textup{cf}} = \hat{\mu}^{(d_1)}_{\textup{cf}} - \hat{\mu}^{(d_2)}_{\textup{cf}}$.
\end{algorithmic}
\end{algorithm}

The key insight is that units in $\mathcal{T}_1$ have no 2-hop neighbors in $\mathcal{V}_2$ (and vice versa), ensuring that each unit's predicted outcome comes from a model trained on independent data. The balanced min-cut in Step 3 can be solved using spectral partitioning \citep{fiedler1973algebraic} or the Kernighan-Lin algorithm \citep{kernighan1970efficient}.

\subsection{More on exposure mappings}\label{sec:exposure_mapping}

\subsubsection{Choices of exposure mappings}

In this section, we discuss some common choices of exposure mappings. 
\begin{example}[SUTVA baseline]
    One common choice is to use the SUTVA baseline, which is assuming there is no interference, and setting $D_i = A_i$. However, this is often unrealistic in practice as it ignores the possibility of spillover effects.
\end{example}

\begin{example}[Total number of treated neighbors]
    One choice is to use the total number of treated neighbors, coupled with a self-treatment indicator, which is given by $D_i = (\sumj \bG_{ij}A_j, A_i)$. This is a continuous exposure mapping, and it is more sensitive to the number of treated neighbors. However, it might be hard to estimate the spillover effect for certain contrasts given that if certain units have a lot of neighbors. 
\end{example}

\begin{example}[Indicator of whether any neighbor is treated]
    One choice is to use the indicator of whether any neighbor is treated, coupled with a self-treatment indicator, which is given by $D_i = (A_i, \indicator{\exists j \in \cN_i : A_j = 1})$. The indicator gives a binary level of spillover, and it is easy to interpret. However, it is not very sensitive to the number of treated neighbors.
\end{example}

\begin{example}[Exposure bins]
    Another common choice is to use exposure bins, Discretize neighbor summary into levels to define a small set of exposures: $S_i \in \{0,1,2,3+\}$, or quantile bins of treated-neighbor proportion. This is a discrete exposure mapping, and it is easy to interpret. 
\end{example}

\subsubsection{Parameter interpretation with misspecified exposure mappings}

We can also follow the framework of \citet{leung2022causal} and allow for misspecified exposure mappings: $D_i$ need not correctly capture how others affect an individual’s potential
outcome. With the potential outcomes $Y_i(\ba)$, we can define the unit $i$'s expected response
under exposure mapping value $d$ as
\begin{align*}
    \mu_i(d) = \sum_{\ba' \in \{0,1\}^n} Y_i(\ba') P(D = \ba' | T = d).
\end{align*}
which equals the expected potential outcome of unit $i$ over all possible treatment assignment vectors given the exposure mapping value at $d$. Let $\mu(d) = n^{-1} \sumi \mu_i(d)$ be the finite population average and $\mu = (\mu(d) : d \in \cD)$ be the $K\times 1$ vector containing all the $\mu(d)$'s corresponding to exposure mapping values. Our focus will be on contrasting these aggregated avarege potential outcomes. 

\subsubsection{Learning exposure mappings via GNNs}
In many applications, the exposure mapping is not known a priori. We propose a data-driven approach to learn the exposure mapping using GNNs. The key idea is to learn a categorical representation that captures how a unit's own treatment and its neighbors' treatments jointly affect the outcome.

\begin{enumerate}[leftmargin=*]
    \item \textbf{Covariate encoding.} Encode pretreatment covariates and network structure: $\phi_i = \textup{GNN}_X(G, \bX)_i$, where $\phi_i \in \mathbb{R}^{d_\phi}$ captures unit $i$'s covariates and local network structure.
    
    \item \textbf{Treatment encoding.} Encode the treatment assignment pattern via a separate GNN: $\psi_i = \textup{GNN}_A(G, \bA)_i$, where $\psi_i \in \mathbb{R}^{d_\psi}$ captures how unit $i$'s own treatment and its neighbors' treatments jointly determine the exposure.
    
    \item \textbf{Exposure categorization.} Pass the treatment encoding through a softmax layer: $\pi_i = \textup{softmax}(W \psi_i)$, where $W \in \mathbb{R}^{K \times d_\psi}$ and $K$ is the number of exposure categories. The discrete exposure is $D_i = \argmax_k \{\pi_{i,k}\}$.
    
    \item \textbf{Gumbel-softmax relaxation.} Since $\argmax$ is non-differentiable, we use the Gumbel-softmax trick \citep{jang2016categorical, maddison2016concrete} during training:
    \begin{align*}
        e_i = \textup{softmax}\left(\frac{\log \pi_i + g_i}{\tau}\right),
    \end{align*}
    where $g_i$ contains i.i.d.\ Gumbel(0,1) samples and $\tau > 0$ is a temperature parameter annealed during training.
    
    \item \textbf{Outcome model and loss.} The outcome model is $\hat{Y}_i = f(\phi_i, e_i; \theta)$, trained by minimizing:
    \begin{align}\label{eq::learned_exposure_loss}
        \mathcal{L} = \underbrace{\frac{1}{n}\sumi (Y_i - \hat{Y}_i)^2}_{\text{MSE loss}} + \lambda \cdot \underbrace{\textup{IPM}(\{\phi_i\}_{D_i=d}, \{\phi_i\}_{D_i=d^{\prime}})}_{\text{Balance loss}},
    \end{align}
    where the IPM loss (e.g., MMD or Wasserstein distance) encourages balanced representations across exposure groups.
\end{enumerate}

At inference time, we use the hard assignment $D_i = \argmax_k \{\pi_{i,k}\}$ to obtain discrete exposure categories for use in our GAUGER framework.

\section{Additional Theoretical Results}
In this section, we provide additional theoretical results. They serve as intermediate steps in the proofs of the main results, but also stand alone as insights on the general properties of the proposed framework. Let's define $\hat\tau{}^{(d_1, d_2)}_{\beta^*}$ as the AIPW estimator \eqref{eq::aipw_def} with a given set of features and the optimal oracle coefficients $\beta^*$ from \eqref{eq::beta-star}. The following theorem shows the consistency of $\hat\tau{}^{(d_1, d_2)}_{\beta^*}$.
\begin{theorem}[Consistency of the linearly adjusted estimator]\label{thm:consistency}
    Under Assumption \ref{asp::positivity_a}-\ref{asp::stability}, the variance of the linearly adjusted estimator has the order:
    \begin{eqnarray}\label{eq::variance_rate}
        \var\{\hat\tau{}^{(d_1, d_2)}_{\beta^*}\} = \Theta\left(\frac{\kappa_n}{n}\right).
    \end{eqnarray}
    Therefore, the linearly adjusted estimator is consistent under the scaling $\kappa_n = o(n)$:
    \begin{align*}
        \hat\tau{}^{(d_1, d_2)}_{\beta^*} - \tau^{(d_1, d_2)} = O_{p}\left(\sqrt{\frac{\kappa_n}{n}}\right) = o_p(1). 
    \end{align*}
    
\end{theorem}

Below we prove asymptotic normality of the AIPW estimator with the oracle coefficients $\beta^*$. 
\begin{theorem}[Asymptotic normality of the linearly adjusted estimator]\label{thm:normality-app}
    Under Assumption \ref{asp::positivity_a}-\ref{asp::stability}, we have the following Berry-Esseen bound:
    \begin{align*}
        \left|P\left(\frac{\hat\tau{}^{(d_1, d_2)}_{\beta^*} - \tau^{(d_1, d_2)}}{\sqrt{\var\{\hat\tau{}^{(d_1, d_2)}_{\beta^*}\}}} \leq z\right) - \Phi(z)\right| \leq \frac{C M_n^3 \delta_n^{20} \eta_n^3}{n^{1/2} \kappa_n^{3/2} },
    \end{align*}
    where $C$ is a constant. Therefore, the linearly adjusted estimator is asymptotically normal:
    \begin{align*}
       \frac{\hat\tau{}^{(d_1, d_2)}_{\beta^*} - \tau^{(d_1, d_2)}}{\sqrt{\var\{\hat\tau{}^{(d_1, d_2)}_{\beta^*}\}}} \rightarrow \cN(0, 1).
    \end{align*}
\end{theorem}
Theorem \ref{thm:normality} provides the asymptotic distribution of the estimator. It builds the foundation of confidence interval construction. The Berry-Esseen bound is of independent interest as an addition to finite-population distribution approximation theory \citep{shi2025berry,Masoero2026multiple}.

Now we need to move to consider the estimator with a plug-in coefficient, as we can never implement the AIPW estimator with an oracle coefficient:
\begin{eqnarray*}
    \hat\tau{}^{(d_1, d_2)}_{\hat\beta} &=& \hat\mu{}^{(d_1)}_{\hat\beta_1} - \hat\mu{}^{(d_2)}_{\hat\beta_2},
\end{eqnarray*}
where, for $k \in \{1,2\}$,
\begin{eqnarray*}
    \hat\mu^{(d_k)}_{\hat\beta} &=& n^{-1}\sumi\frac{\indicator(D_i=d_k)(Y_i-F_i^{(d_k)\T} \hat\beta_k)}{\pi_i^{(d_k)}} + F_i^{(d_k)\T} \hat\beta_k.
\end{eqnarray*}
with $\hat\beta_1$ and $\hat\beta_2$ being the plug-in coefficients.  
\begin{theorem}[Properties with a plug-in coefficient]\label{thm:plug_in}
    Under Assumption \ref{asp::positivity_a}-\ref{asp::stability}, and assume that $\hat\beta$ is consistent, i.e., $\hat\beta - \beta^\star = o_p(1)$. The plug-in coefficient estimator is asymptotically normal as $n\to\infty$:
    \begin{eqnarray*}
        \frac{\hat\tau^{\textup{adj}}_{\hat{\beta}} - \tau}{\sqrt{\avar\{\hat\tau^{\textup{adj}}_{\hat\beta}\}}} \rightarrow \cN(0, 1).
    \end{eqnarray*}
\end{theorem}
Theorem \ref{thm:plug_in} is a strong result stating that the plug-in AIPW estimator also enjoys the asymptotic normality, as long as the estimated coefficient is consistent. Hence, the focus of the problem goes to verifying the GER procedure from \eqref{eq::regression} gives a consistent estimator to the oracle $\beta^\star$. This is established by Theorem \ref{thm:convergence_beta}.


\section{Proof}\label{sec:proof}
\subsection{Proof of Theorem \ref{thm:mean_variance}}
\begin{proof}
    First, we have that
    \begin{eqnarray*}
        \E\{(Y_i - F_i^{(d_k)\T} \beta_k)\indicator(D_i=d_k)/\pi_i^{(d_k)}\} = Y_i^{(d_k)} - F_i^{(d_k)\T} \beta_k.
    \end{eqnarray*}
    Taking a sum over all units, we have that
    \begin{eqnarray*}
        \E\{\hat\mu{}^{(d_1)}_{\beta_1}\} = n^{-1} \sumi Y_i^{(d_1)} - n^{-1} \sumi F_i^{(d_1)\T} \beta_1 + n^{-1} \sumi F_i^{(d_1)\T} \beta_1 = \mu^{(d_1)}, \\
        \E\{\hat\mu{}^{(d_2)}_{\beta_2}\} = n^{-1} \sumi Y_i^{(d_2)} - n^{-1} \sumi F_i^{(d_2)\T} \beta_2 + n^{-1} \sumi F_i^{(d_2)\T} \beta_2 = \mu^{(d_2)}.
    \end{eqnarray*} 
    Thus for any $\beta$, we have that
    \begin{eqnarray*}
        \E\{\hat\tau^{(d_1, d_2)}_{\beta}\} = \E\{\hat\mu^{(d_1)}_{\beta_1}\} - \E\{\hat\mu^{(d_2)}_{\beta_2}\} = \tau^{(d_1, d_2)}.
    \end{eqnarray*}

    For variance, because the only random component in the estimators is the exposure levels $D_i$. For unit $i$, we have:
    \begin{eqnarray*}
        \cov(\indicator(D_i=d), \indicator(D_i=d)) &=& \E(\indicator(D_i=d) \cdot \indicator(D_i=d)) - \E(\indicator(D_i=d)) \cdot \E(\indicator(D_i=d)) \\
        &=& \pi_i^{(d)}(1 - \pi_i^{(d)}), \\
        \cov(\indicator(D_i=d_1), \indicator(D_i=d_2)) &=& \E(\indicator(D_i=d_1) \cdot \indicator(D_i=d_2)) - \E(\indicator(D_i=d_1)) \cdot \E(\indicator(D_i=d_2)) \\
        &=& 0 - \pi_i^{(d_1)}\pi_i^{(d_2)} = - \pi_i^{(d_1)}\pi_i^{(d_2)}, \quad \text{for } d_1\neq d_2.   
    \end{eqnarray*}    
    For two units $i$ and $j$, we have:
    \begin{eqnarray*}
        \cov(\indicator(D_i=d), \indicator(D_j=d)) &=& \E(\indicator(D_i=d) \cdot \indicator(D_j=d)) - \E(\indicator(D_i=d)) \cdot \E(\indicator(D_j=d)) \\
        &=& \pi_{ij}^{(d,d)} - \pi_i^{(d)}\pi_j^{(d)}, \\
        \cov(\indicator(D_i=d_1), \indicator(D_j=d_2)) &=& \pi_{ij}^{(d_1,d_2)} - \pi_i^{(d_1)}\pi_j^{(d_2)} . 
    \end{eqnarray*}    
    The variance of $\hat\tau^{(d_1, d_2)}_{\beta}$ consists of three parts:
    \begin{eqnarray*}
        \var\{\hat\tau^{(d_1, d_2)}_{\beta}\} &=& \var\{\hat\mu^{(d_1)}_{\beta_1}\} + \var\{\hat\mu^{(d_2)}_{\beta_2}\} - 2\cov\{\hat\mu^{(d_1)}_{\beta_1}, \hat\mu^{(d_2)}_{\beta_2}\}.
    \end{eqnarray*}
    For $\var\{\hat\mu^{(d_1)}_{\beta_1}\}$, we have:
    \begin{eqnarray*}
        \var\{\hat\mu^{(d_1)}_{\beta_1}\} &=& \var\left(\frac{1}{n}\sumi \left\{ \frac{\indicator(D_i=d_1)(Y_i-F_i^{(d_1)\T} \beta_1)}{\pi_i^{(d_1)}} + F_i^{(d_1)\T}\beta_1 \right\}\right) \\
        &=& \var\left(\frac{1}{n}\sumi \left\{ \frac{\indicator(D_i=d_1)(Y_i-F_i^{(d_1)\T} \beta_1)}{\pi_i^{(d_1)}} \right\}\right)\\
        &=& \frac{1}{n^2} \sumi \frac{\var(\indicator(D_i = d_1)) \cdot (Y_i^{(d_1)} - F_i^{(d_1)\T}\beta_1)^2}{\{\pi_i^{(d_1)}\}^2}\\
        &&  + \frac{1}{n^2} \sumij \frac{\cov(\indicator(D_i=d_1), \indicator(D_j=d_1)) (Y_i^{(d_1)} - F_i^{(d_1)}\beta_1)^{\T} \beta_1)(Y_j^{(d_1)} - F_j^{(d_1)}\beta_1)^{\T} \beta_1)}{\pi_i^{(d_1)}\pi_j^{(d_1)}} \\
        &=& \frac{1}{n^2} \{\bY^{(d_1)}-\bF^{(d_1)}\beta_1\}^{\T} \bLambda_{11}\{\bY^{(d_1)}-\bF^{(d_1)}\beta_1\}. 
    \end{eqnarray*}
    Similarly, we can show that
    \begin{eqnarray*}
        \var\{\hat\mu^{(d_2)}_{\beta_2}\} &=& \frac{1}{n^2} \{\bY^{(d_2)}-\bF^{(d_2)}\beta_2\}^{\T} \bLambda_{22}\{\bY^{(d_2)}-\bF^{(d_2)}\beta_2\}. 
    \end{eqnarray*}
    For $\cov\{\hat\mu^{(d_1)}_{\beta_1}, \hat\mu^{(d_2)}_{\beta_2}\}$, we have:
    \begin{eqnarray*}
        \cov\{\hat\mu^{(d_1)}_{\beta_1}, \hat\mu^{(d_2)}_{\beta_2}\} 
        &=& \frac{1}{n^2} \{\bY^{(d_1)}-\bF^{(d_1)}\beta_1\}^{\T} \bLambda_{12}\{\bY^{(d_2)}-\bF^{(d_2)}\beta_2\}. 
    \end{eqnarray*}
    Thus, we conclude the theorem. 
\end{proof}

\subsection{Statement and Proof of Lemma \ref{lem::kappa_n_upper_bound}}
\begin{lemma}\label{lem::kappa_n_upper_bound}
    Under Assumption \ref{asp::positivity_a}-\ref{asp::stability}, we have that
    \begin{eqnarray*}
        \kappa_n \leq {C\eta_n^2 \delta_n M_n^2} = O(\operatorname{poly}(\log n)). 
    \end{eqnarray*}
\end{lemma}
\begin{proof}
    Note that we can bound the largest eigenvalue of $\bLambda$ with its induced $\infty$-norm:
    \begin{eqnarray}\label{eq::lambda_max_bound}
        \rho_{\max}(\bLambda) \le \|\bLambda\|_{\infty} = \max_{i\in[n]} \sum_{j\in[n]} |(\bLambda_{12})_{(i,j)}| \le C\eta_n^2 \delta_n.
    \end{eqnarray} 
    Based on Assumption \ref{asp::stability}, we have that
    \begin{eqnarray*}
    ({n\kappa_n})^{-1} \rho_{\min}\left\{ \begin{pmatrix} \bY(d_1)^{\T} & 0 \\ 0 & \bY(d_2)^{\T} \end{pmatrix} \bLambda(d_1,d_2) \begin{pmatrix} \bY(d_1) & 0 \\ 0 & \bY(d_2) \end{pmatrix} \right\} \ge \underline{c}.
    \end{eqnarray*}
    Moreover,
    \begin{eqnarray*}
        \rho_{\max}\left\{ \begin{pmatrix} \bY(d_1)^{\T} & 0 \\ 0 & \bY(d_2)^{\T} \end{pmatrix} \bLambda(d_1,d_2) \begin{pmatrix} \bY(d_1) & 0 \\ 0 & \bY(d_2) \end{pmatrix} \right\} \le C\eta_n^2 \delta_n (\|\bY(d_1)\|_2^2 + \|\bY(d_2)\|_2^2).
        \end{eqnarray*}
    Therefore, we have that
    \begin{eqnarray*}
        ({n\kappa_n})^{-1} \ge \frac{\underline{c}}{C\eta_n^2 \delta_n (\|\bY(d_1)\|_2^2 + \|\bY(d_2)\|_2^2)},
    \end{eqnarray*}
    or equivalently,
    \begin{eqnarray*}
        \kappa_n \leq \frac{C\eta_n^2 \delta_n (\|\bY(d_1)\|_2^2 + \|\bY(d_2)\|_2^2)}{n\underline{c}} \leq {C\eta_n^2 \delta_n M_n^2}.
    \end{eqnarray*}
\end{proof}

\subsection{Proof of Theorem \ref{thm:consistency}}
\begin{proof}
    We will use Chebyshev's inequality to show convergence in probability. Under Assumption \ref{asp::stability}, for $\beta^*$, based on Equation \eqref{eq::variance_rate}, we have that
    \begin{align}
        \frac{2\underline{c}\kappa_n}{n} \precsim \var\{\hat\tau^{(d_1, d_2)}_{\beta^*}\} = \frac{1}{n^2}
        \begin{pmatrix} 1 & -1 \end{pmatrix}
        \bSigma^\varepsilon_n
        \begin{pmatrix} 1 \\ -1 \end{pmatrix} \precsim \frac{2\overline{c}\kappa_n}{n} .
    \end{align}
    Applying the Chebyshev's inequality, we have that
    \begin{eqnarray*}
        \hat\tau^{(d_1, d_2)}_{\beta^\star} - \tau^{(d_1, d_2)} = O_{p}\left(\sqrt{\frac{\kappa_n}{n}}\right) = o_p(1).
    \end{eqnarray*}                 
\end{proof}

\subsection{Proof of Theorem \ref{thm:normality-app} and Theorem \ref{thm:normality}}

Theorem \ref{thm:normality} is a direct corollary of the more general result Theorem \ref{thm:normality-app}. Hence, we directly prove Theorem \ref{thm:normality-app}. 

\paragraph{A preliminary lemma.} We will use a result from \cite{chen2004normal}. For self-consistency, we include the following lemma for CLT under graph dependency. Consider a set of random variables $\{V_i, i \in \mathcal{V}\}$ indexed by the vertices
of a graph $\mathcal{G} = (\mathcal{V}, \mathcal{E})$. $\mathcal{G}$ is said to be a dependency graph if, for any pair of disjoint sets $\Gamma_1$ and $\Gamma_2$ in $\mathcal{V}$ such that no edge in $\mathcal{E}$ has one endpoint in $\Gamma_1$ and the other in $\Gamma_2$, the sets of random variables $\{V_i, i \in \Gamma_1\}$ and $\{V_i, i \in \Gamma_2\}$ are independent. Let $D$ denote the maximal degree of $\mathcal{G}$, that is, the maximal number of edges incident to a single vertex.
\begin{lemma}[A CLT under graph dependency, Theorem 2.7 in \cite{chen2004normal}]\label{lem::clt_graph_dependency}
    Let $\{V_i, i \in \mathcal{V}\}$ be random variables indexed by the vertices of a dependency graph $\mathcal{G} = (\mathcal{V}, \mathcal{E})$. Put $W = \sumi V_i$. Assume that $\E(W^2) = 1$, $\E(V_i) = 0$ and $\E(|V_i|^p) \leq \theta^p$ for $i \in \mathcal{V}$ and for some $2 < p \le 3$ and $\theta > 0$. Then for any $z\in\mathbb{R}$,
    \begin{eqnarray*}
        |P(W \leq z) - \Phi(z)| \leq 75 D^{5(p-1)}|V|\theta^p. 
    \end{eqnarray*}
\end{lemma}

\begin{proof}
    The proof follows from the central limit theorem. We will use Lemma \ref{lem::clt_graph_dependency} to show a CLT for $\{\hat\mu^{(d_k)}_{\beta_k}\}_{k\in[2]}$. 
    
    Taking $V_i$ in Lemma \ref{lem::clt_graph_dependency} as the scaled random component, dropping the outcome part, we have that
    \begin{eqnarray*}
        V_i = \frac{1}{n\sqrt{\var\{\hat\mu^{(d_k)}_{\beta_k}\}}} \left\{\frac{\indicator(D_i=d_k)(Y_i - F_i^{(d_k)\T} \beta_k)}{\pi_i^{(d_k) } } - (Y_i^{(d_k)} - F_i^{(d_k)\T} \beta_k)\right\}, 
    \end{eqnarray*}
    with $\E(V_i) = 0$. Also define 
    \begin{eqnarray*}
        W = \sumi V_i, \text{ with } \E(W^2) = 1.
    \end{eqnarray*}
    Now we bound $\E(|V_i|^3)$. By Assumption \ref{asp::bound}, we have that
    \begin{eqnarray*}
        &&\E\left(\left|\frac{\indicator(D_i=d_1)(Y_i - F_i^{(d_1)\T} \beta_1)}{\pi_i^{(d_1) } } - (Y_i^{(d_1)} - F_i^{(d_1)\T} \beta_1)\right|^3\right) \\
        &\le& \left(|Y_i^{(d_1)} - F_i^{(d_1)\T} \beta_1|^3\right) \cdot \E\left(\left|\frac{\indicator(D_i=d_1)}{\pi_i^{(d_1)}} - 1\right|^3\right) \\
        &\le& C M_n^3 \eta_n^3.
    \end{eqnarray*}
    Similarly, we have that
    \begin{eqnarray*}
        &&\E\left(\left|\frac{\indicator(D_i=d_2)(Y_i - F_i^{(d_2)\T} \beta_2)}{\pi_i^{(d_2) } } - (Y_i^{(d_2)} - F_i^{(d_2)\T} \beta_2)\right|^3\right) \\
        &\le& \left(|Y_i^{(d_2)} - F_i^{(d_2)\T} \beta_2|^3\right) \cdot \E\left(\left|\frac{\indicator(D_i=d_2)}{\pi_i^{(d_2)}} - 1\right|^3\right) \\
        &\le& C M_n^3 \eta_n^3.
    \end{eqnarray*}
    Therefore, using \eqref{eq::variance_rate}, we have that
    \begin{eqnarray*}
        \E(|V_i|^3) \leq \frac{C M_n^3 \eta_n^3}{n^3 \var\{\hat\mu^{(d_k)}_{\beta_k}\}^{3/2}} \leq \frac{C M_n^3 \eta_n^3}{n^{3/2} \kappa_n^{3/2} } \text{ (defined as $\theta^p$ in Lemma \ref{lem::clt_graph_dependency})}.
    \end{eqnarray*}
    Based on the definition for dependency graph, we have can define $\cG = \bDelta$ to be the dependency graph of the $\bDelta$. The maximal degree of $\cG$ is at most $\delta_n$. Also $|V| = n$. Therefore, using Lemma \ref{lem::clt_graph_dependency}, we have that
    \begin{eqnarray*}
        |P(W \leq z) - \Phi(z)| \leq \frac{C M_n^3 \delta_n^{10} \eta_n^3}{n^{1/2} \kappa_n^{3/2} }. 
    \end{eqnarray*}

\end{proof}

\subsection{Proof of Theorem \ref{thm:plug_in}}
\begin{proof}
    We compute that
    \begin{eqnarray*}
        \hat\tau^{(d_1, d_2)}_{\beta^*} - \hat\tau^{(d_1, d_2)}_{\hat\beta} &=& \frac{1}{n}\sumi \left\{ \frac{\indicator(D_i=d_1)\{-F_i^{(d_1)\T} (\beta_1^* - \hat\beta_1)\}}{\pi_i^{(d_1)}} + F_i^{(d_1)\T}(\beta_1^* - \hat\beta_1) \right\} \\
        &&- \frac{1}{n}\sumi \left\{ \frac{\indicator(D_i=d_2)\{-F_i^{(d_2)\T} (\beta_2^* - \hat\beta_2)\}}{\pi_i^{(d_2)}} + F_i^{(d_2)\T}(\beta_2^* - \hat\beta_2) \right\} \\
        &=& -(\beta_1^* - \hat\beta_1)\cdot \frac{1}{n}\sumi  \left\{\frac{\indicator(D_i=d_1)}{\pi_i^{(d_1)}} - 1\right\}F_i^{(d_1)\T} \\
        &&+ (\beta_2^* - \hat\beta_2)\cdot \frac{1}{n}\sumi \left\{\frac{\indicator(D_i=d_2)}{\pi_i^{(d_2)}} - 1\right\}F_i^{(d_2)\T}.
    \end{eqnarray*}
    Treating $F_i^{(d_k)}$ as pseudo potential outcomes, we can apply Theorem \ref{thm:consistency} to obtain rate:
    \begin{eqnarray*}
        \frac{1}{n}\sumi  \left\{\frac{\indicator(D_i=d_1)}{\pi_i(d_1)} - 1\right\}F_i^{(d_1)\T} = O_p\left(\sqrt{\frac{\kappa_n}{n}}\right), \quad \frac{1}{n}\sumi \left\{\frac{\indicator(D_i=d_2)}{\pi_i(d_2)} - 1\right\}F_i^{(d_2)\T} = O_p\left(\sqrt{\frac{\kappa_n}{n}}\right).
    \end{eqnarray*}
    Therefore,  Using $\hat\beta - \beta^* = o_p(1)$ and the variance order \eqref{eq::variance_rate}, we have
    \begin{eqnarray*}
        \frac{\hat\tau^{(d_1, d_2)}_{\beta^*} - \hat\tau^{(d_1, d_2)}_{\hat\beta}}{\sqrt{\var(\hat\tau^{(d_1, d_2)}_{\beta^*})}} = \frac{O_p\left(\sqrt{\frac{\kappa_n}{n}}\right) \cdot o_p(1)}{\Theta\left(\sqrt{\frac{\kappa_n}{n}}\right)}  = o_p(1).
    \end{eqnarray*}
    Thus we conclude the theorem.
\end{proof}

\subsection{Statement and Proof of Lemma \ref{lem::variance_bounds}}
\begin{lemma}[Variance bounds for a double-sum statistic over a dependency graph]\label{lem::variance_bounds}
    Let $\cG = (\mathcal{V}, \mathcal{E})$ be a dependency graph, with adjacency matrix $\bDelta$ and maximum degree $\delta(\bDelta)$. Let $(U_i, V_i)_{i\in\mathcal{V}}$ be a sequence of random variable pairs following $\cG$, with mean zero and variance $\max\{\E(U_i^4), \E(V_i^4)\} \le \bar{\sigma}^4$. Let $(c_i)_{i\in\mathcal{V}}$ be a sequence of constants, with upper bound $\bar{c}$. Then,
    \begin{eqnarray*}
        \var\left\{\sumij (\bDelta)_{ij} \cdot c_{ij} U_i V_j\right\} \leq n \delta(\bDelta)^3 \bar{c}^2 \bar{\sigma}^4.
    \end{eqnarray*}
\end{lemma}
\begin{proof}
    We have that
    \begin{eqnarray*}
        \var\left\{\sumij (\bDelta)_{ij} \cdot c_{ij} U_i V_j\right\} &=& \sumij \sumrs (\bDelta)_{ij} (\bDelta)_{rs} \cdot c_{ij} c_{rs} \cov(U_iV_j, U_rV_s).
    \end{eqnarray*}
    Here, $\cov(U_iV_j, U_rV_s) = 0$ if $U_iV_j \indep U_rV_s$. Therefore, the index tuple $(i,j,r,s)\in[n]^4$ for the nonzero terms in the summation must satisfy:
    \begin{eqnarray*}
        U_i \nonindep V_j, \quad U_r \nonindep V_s, \quad U_iV_j \nonindep U_rV_s.
    \end{eqnarray*}
    On dependency graph, this means: $i$ is connected to $j$, $r$ is connected to $s$, and at least one of $i,j$ is connected to one of $r,s$. WLOS, if $j$ connects to $r$, this means we can find a path on the dependency graph in the form of $i\to j\to r \to s$ (some index might be the same, say $i = j$, as dependency graph includes self-loop). Conversely, any path with length no greater than three can lead to a potential nonzero index tuple. For each node $i$, there are at most $\delta(\bDelta)^3$ potential paths with length no greater than three. Therefore, the total number of potential nonzero index tuples is at most $n \delta(\bDelta)^3$.

    Now each term in the summation is upper bounded by 
    \begin{eqnarray*}
        |c_{ij} c_{rs} \cov(U_iV_j, U_rV_s)| \leq \bar{c}^2 \sqrt{\var(U_iV_j) \var(U_rV_s)} \leq \bar{c}^2 \sqrt{(\E(U_i^4)+\E(V_j^4))(\E(U_r^4)+\E(V_s^4))/4} \le \bar{c}^2 \bar{\sigma}^4.
    \end{eqnarray*}
    Therefore, we conclude that 
    \begin{eqnarray*}
        \var\left\{\sumij (\bDelta)_{ij} \cdot c_{ij} U_i V_j\right\} \leq n \delta(\bDelta)^3 \bar{c}^2 \bar{\sigma}^2.
    \end{eqnarray*}
    
\end{proof}

\subsection{Proof of Theorem \ref{thm:convergence_beta}}
\begin{proof}
    For convience, recall the notaion
    \begin{eqnarray*}
        \bFaug = \begin{pmatrix} \bF^{(d_1)} & \bZero_{n\times p} \\ \bZero_{n\times p} & \bF^{(d_2)} \end{pmatrix}.
    \end{eqnarray*}
    Let $\hat\beta^{\textup{regression}}$ be the plug-in coefficients from \eqref{eq::regression}, with the pseudo inverse trimmed to the rank $r$. Let's show that 
    \begin{eqnarray*}
        (n\kappa_n)^{-1} 
        \begin{pmatrix}
            \hat\bF{}^{(d_1)\T} - \bF^{(d_1)\T} \\
            \hat\bF{}^{(d_2)\T} - \bF^{(d_2)\T}
        \end{pmatrix} 
        \bDelta 
        \begin{pmatrix}
            \hat\bF{}^{(d_1)\T} - \bF^{(d_1)\T} & 
            \hat\bF{}^{(d_2)\T} - \bF^{(d_2)\T}
        \end{pmatrix}
    \end{eqnarray*}
    converges in probability to its expectation. Taking the left-upper block, for any $u_{k,p}, u_{l,p}\in\mathbb{R}^{p}$ with all zero but a 1 at the $k$-th and $l$-th positions, we have that
    \begin{eqnarray*}
        &&(n\kappa_n)^{-1} 
        u_{k,p}^{\T}(\hat\bF{}^{(d_1)\T} - \bF^{(d_1)\T})
        \bDelta 
        (\hat\bF{}^{(d_1)} - \bF^{(d_1)})u_{l,p} \\
        &=& (n\kappa_n)^{-1} \sumij \indicator(D_i \nonindep D_j) \cdot F_i^{(d_1)\T}u_{k,p} F_j^{(d_1)\T}u_{l,p} \cdot (\hat{1}_i(d_1) - 1)(\hat{1}_j(d_1) - 1).
    \end{eqnarray*}
    Taking expectation, we have that
    \begin{eqnarray*}
        &&\E\{(n\kappa_n)^{-1} 
        u_{k,p}^{\T}(\hat\bF{}^{(d_1)\T} - \bF^{(d_1)\T})
        \bDelta 
        (\hat\bF{}^{(d_1)} - \bF^{(d_1)}) u_{l,p}\} \\
        &=& (n\kappa_n)^{-1} \sumij \indicator(D_i \nonindep D_j) \cdot F_i^{(d_1)\T}u_{k,p} F_j^{(d_1)\T}u_{l,p} \cdot \E\{(\hat{1}_i(d_1) - 1)(\hat{1}_j(d_1) - 1)\}\\
        &=& (n\kappa_n)^{-1} \sumij \indicator(D_i \nonindep D_j) \cdot F_i^{(d_1)\T}u_{k,p} F_j^{(d_1)\T}u_{l,p} \cdot \cov(\indicator(D_i=d_1), \indicator(D_j=d_1)) \\
        &=& (n\kappa_n)^{-1} \sumij \indicator(D_i \nonindep D_j) \cdot F_i^{(d_1)\T}u_{k,p} F_j^{(d_1)\T}u_{l,p} \cdot \{\bLambda_{11}\}_{ij} \\
        &=& (n\kappa_n)^{-1} \cdot u_{k,p}^{\T} \bF^{(d_1)\T} \bLambda_{11} \bF^{(d_1)} u_{l,p}. 
    \end{eqnarray*}
    Using Lemma \ref{lem::variance_bounds} and Lemma \ref{lem::kappa_n_upper_bound}, we have that
    \begin{eqnarray*}
        \var\left\{(n\kappa_n)^{-1} 
        u_{k,p}^{\T}(\hat\bF{}^{(d_1)\T} - \bF^{(d_1)\T})
        \bDelta 
        (\hat\bF{}^{(d_1)} - \bF^{(d_1)}) u_{l,p}\right\} \leq \frac{C n \delta_n^6 M_n^4 \eta_n^6}{n^2 \kappa_n^2} = O(\frac{C  \delta_n^6 M_n^4 \eta_n^6}{n \kappa_n^2}) = o(1).
    \end{eqnarray*}
    Therefore, we have that
    \begin{eqnarray*}
        (n\kappa_n)^{-1}\{(\hat\bF{}^{(d_1)\T} - \bF^{(d_1)\T})
        \bDelta 
        (\hat\bF{}^{(d_1)} - \bF^{(d_1)}) - \bF^{(d_1)\T} \bLambda_{11} \bF^{(d_1)}\} = O_p(\sqrt{\frac{\eta_n^6 \delta_n^3 M_n^4}{n \kappa_n^2}}).
    \end{eqnarray*}  
    Similarly, we can show that this convergence holds for all the other three blocks, which amounts to 
    \begin{eqnarray*}
        (n\kappa_n)^{-1}\{(\hat\bF{}^{(d_1; d_2)\T} - \bF^{(d_1; d_2)\T})
        \bDelta 
        (\hat\bF{}^{(d_1; d_2)} - \bF^{(d_1; d_2)}) - \bFaug{}^\T \bLambda \bFaug\} = O_p(\sqrt{\frac{\eta_n^6 \delta_n^3 M_n^4}{n \kappa_n^2}}).
    \end{eqnarray*}

    Using Lemma \ref{lem::optimal_perturbation_bounds}, we have that 
    \begin{eqnarray*}
        &&(n\kappa_n)\|\{(\hat\bF{}^{(d_1; d_2)\T} - \bF^{(d_1; d_2)\T})
        \bDelta 
        (\hat\bF{}^{(d_1; d_2)} - \bF^{(d_1; d_2)})\}^{\dagger} - \{\bFaug{}^\T \bLambda \bFaug\}^{\dagger} \|_{\textup{op}} \\
        &\le& C \max\{\|(n\kappa_n)\{(\hat\bF{}^{(d_1; d_2)\T} - \bF^{(d_1; d_2)\T})
        \bDelta 
        (\hat\bF{}^{(d_1;d_2)} - \bF^{(d_1;d_2)})\}^{\dagger}\|_{\textup{op}}^2, \|(n\kappa_n)\{(\bFaug)^{\T} \bLambda\bFaug\}^{\dagger}\|_{\textup{op}}^2\} \\ && \cdot \|(n\kappa_n)^{-1}\{(\hat\bF{}^{(d_1;d_2)\T} - \bF^{(d_1;d_2)\T})
        \bDelta 
        (\hat\bF{}^{(d_1;d_2)} - \bF^{(d_1;d_2)}) - \bFaug{}^\T \bLambda \bFaug\}\|_{\textup{op}}.
    \end{eqnarray*}
    According to Assumption \ref{asp::stability}, we have
    \begin{eqnarray*}
        \max\{\|(n\kappa_n)\{(\hat\bF{}^{(d_1; d_2)\T} - \bF^{(d_1; d_2)\T})
        \bDelta 
        (\hat\bF{}^{(d_1; d_2)} - \bF^{(d_1; d_2)})\}^{\dagger}\|_{\textup{op}}^2, \|(n\kappa_n)\{(\bFaug)^{\T} \bLambda\bFaug\}^{\dagger}\|_{\textup{op}}^2\} \le C \underline{c}^{-2}.
    \end{eqnarray*}
    Therefore, we have that 
    \begin{eqnarray*}
        (n\kappa_n)\|((\hat\bF{}^{(d_1; d_2)\T} - \bF^{(d_1; d_2)\T})
        \bDelta 
        (\hat\bF{}^{(d_1; d_2)} - \bF^{(d_1; d_2)}))^{\dagger} - (\bFaug{}^\T \bLambda \bFaug)^{\dagger} \|_{\textup{op}} = O_p(\sqrt{\frac{\eta_n^6 \delta_n^3 M_n^4}{n \kappa_n^2}}).
    \end{eqnarray*}
    Moreover, we have that 
    \begin{eqnarray*}
        (n\kappa_n)^{-1}\{(\hat\bF{}^{(d_1; d_2)\T} - \bF^{(d_1; d_2)\T})
        \bDelta 
        (\hat\bY{}^{(d_1)} - \hat\bY{}^{(d_2)}) - \bFaug{}^\T \bLambda \bY^{(d_1; d_2)}\} = O_p(\sqrt{\frac{\eta_n^6 \delta_n^3 M_n^4}{n \kappa_n^2}}).
    \end{eqnarray*}
    Therefore, we have that
    \begin{eqnarray*}
        \hat\beta - \beta^\star = O_p\left(\sqrt{\frac{\eta_n^6 \delta_n^6 M_n^4}{n \kappa_n^2}}\right).
    \end{eqnarray*}

\end{proof}

\subsection{Proof of Theorem \ref{thm:var_est}}\label{sec:pf-var-est}
We first introduce a lemma that will be used in the proof. 
\begin{lemma}[Perturbation bounds of the Moore-Penrose inverse, \citet{wedin1973perturbation}]\label{lem::optimal_perturbation_bounds}
    Let $\bA$ be a real-valued matrix and $B = A + E$. Then,
    \begin{eqnarray*}
        \|B^\dagger - A^\dagger\|_{\textup{op}} \leq C \max\{\|A^\dagger\|_{\textup{op}}^2, \|B^\dagger\|_{\textup{op}}\} \|E\|_{\textup{op}},
    \end{eqnarray*}
\end{lemma}

Now we are ready to prove Theorem \ref{thm:var_est}. 
\begin{proof}
    We first show the consistency of the variance estimator for $\hat\mu(d_1)$ and $\hat\mu(d_2)$. We focus on $\hat{v}(d_1)$ without loss of generality. Recall that 
    \begin{eqnarray*}
        n\kappa_n^{-1}\hat{v}(d_1) &=& (n\kappa_n)^{-1}\{\hat{\bY}{}^{(d_1)} - \hat \bF{}^{(d_1)}\hat\beta_1\}^\T (\bLambda_{11}^\pi + \bPhi_{11}^\pi) \{\hat{\bY}{}^{(d_1)} - \hat \bF{}^{(d_1)}\hat\beta_1\} \\
        &=& \underbrace{(n\kappa_n)^{-1} \{\hat{\bY}{}^{(d_1)} - \hat \bF{}^{(d_1)}\beta^\star_1\}^\T (\bLambda_{11}^\pi + \bPhi_{11}^\pi) \{\hat{\bY}{}^{(d_1)} - \hat \bF{}^{(d_1)}\beta^\star_1\}}_{\text{Term I}} \\
        &&+ \underbrace{(n\kappa_n)^{-1} \{\hat \bF{}^{(d_1)}\hat\beta_1 - \hat \bF{}^{(d_1)}\beta^\star_1\}^\T (\bLambda_{11}^\pi + \bPhi_{11}^\pi) \{\hat \bF{}^{(d_1)}\hat\beta_1 - \hat \bF{}^{(d_1)}\beta^\star_1\}}_{\text{Term II}} \\
        &&+ \underbrace{2(n\kappa_n)^{-1} \{\hat{\bY}{}^{(d_1)} - \hat \bF{}^{(d_1)}\beta^\star_1\}^\T (\bLambda_{11}^\pi + \bPhi_{11}^\pi) \{\hat \bF{}^{(d_1)}\hat\beta_1 - \hat \bF{}^{(d_1)}\beta^\star_1\}}_{\text{Term III}}. 
    \end{eqnarray*} 
    For Term I, we have that
    \begin{eqnarray*}
        && (n\kappa_n)^{-1} \{\hat{\bY}{}^{(d_1)} - \hat \bF{}^{(d_1)}\beta^\star_1\}^\T (\bLambda_{11}^\pi + \bPhi_{11}^\pi) \{\hat{\bY}{}^{(d_1)} - \hat \bF{}^{(d_1)}\beta^\star_1\}\\
        &= &\underbrace{(n\kappa_n)^{-1} \sumij 1_{\{\bDelta_{ij} = 1, \pi_{ij}^{(d_1, d_2)} > 0\}} \{\bLambda_{11}^\pi\}_{(i,j)} (Y_i^{(d_1)} - F_i^{(d_1)\T}\beta^\star_1)(Y_j^{(d_1)} - F_j^{(d_1)\T}\beta^\star_1) \cdot \frac{\indicator(D_i=d_1, D_j=d_2)}{\pi_{i}^{(d_1)}\pi_{j}^{(d_2)}} }_{\text{Term I.1}} \\
        &&+ \underbrace{(n\kappa_n)^{-1} \sumi \left\{\pi_i^{(d_1)} \phi_{i+}^{(d_1, d_2)}\right\} \frac{(Y_i^{(d_1)} - F_i^{(d_1)\T} \beta^\star_1)^2 \indicator(D_i = d_1)}{\{\pi_i^{(d_1)}\}^2} }_{\text{Term I.2}}  .
    \end{eqnarray*}
    Let $\bLambda_{11}^{+}$ be the $n\times n$ block matrix with $(i,j)$-th item equal to
    \begin{eqnarray*}
        \{\bLambda_{11}^{+}\}_{(i,j)} = \{\bLambda_{11}\}_{(i,j)} \cdot \indicator(\pi_{ij}^{(d_1, d_2)} > 0).
    \end{eqnarray*}
    We can verify that
    \begin{eqnarray*}
        \E\{\text{Term I}\} = (n\kappa_n)^{-1} (\bY^{(d_1)}-\bF^{(d_1)}\beta^\star_1)^{\T} (\bLambda_{11}^{+} + \bPhi_{11}) (\bY^{(d_1)}-\bF^{(d_1)}\beta^\star_1) \triangleq (n\kappa_n)^{-1}v_{\lim}(d_1).
    \end{eqnarray*}
    By the fact that $\bLambda_{11}^{+} + \bPhi_{11} \succeq \bLambda_{11}$, we always have
    \begin{eqnarray*}
        \E\{\text{Term I}\} \ge n\kappa_n^{-1}\var(\hat\mu{}^{(d_1)}_{\beta_1}) = \Theta(1).
    \end{eqnarray*}

    Now, using Lemma \ref{lem::variance_bounds}, we have that
    \begin{eqnarray*}
        \var\{\text{Term I.1}\} \leq  \frac{C\delta_n^6 M_n^4 \eta_n^6}{n\kappa_n^2}.
    \end{eqnarray*}
    For Term I.2, treating $Y^{\text{pseudo}}_i(d_1) = \pi_i^{(d_1)}\phi_{i+}^{(d_1, d_2)} (Y_i^{(d_1)} - F_i^{(d_1)\T} \beta^\star_1)^2$ as a pseudo potential outcome, using Theorem \eqref{thm:mean_variance},  we have that
    \begin{eqnarray*}
        \var\{\text{Term I.2}\} \leq C (n\kappa_n)^{-1} \delta_n^2 M_n^2 \eta_n^2,
    \end{eqnarray*}
    which is dominated by the variance of Term I.1. To sum up, we have
    \begin{eqnarray*}
        \text{Term I} - \E\{\text{Term I}\} = O_p(\sqrt{\frac{\delta_n^6 M_n^4 \eta_n^6}{n\kappa_n^2}}).
    \end{eqnarray*}

    For Term II, we first bound the variance of the following term, using Lemma \ref{lem::variance_bounds} again: 
    \begin{eqnarray*}
        \var\{ (n\kappa_n)^{-1}u_{k,p}^{\T}\hat\bF{}^{(d_1)\T} (\bLambda_{11}^{+} + \bPhi_{11}) \hat\bF{}^{(d_1)}u_{l,p} \} \leq  \frac{C\delta_n^6 M_n^4 \eta_n^6}{n\kappa_n^2}.
    \end{eqnarray*}
    Hence, $(n\kappa_n)^{-1}\hat\bF{}^{(d_1)\T} (\bLambda_{11}^{+} + \bPhi_{11}) \hat\bF{}^{(d_1)} = O_p(1)$, i.e., has the same order as its mean.  
    Based on the stated condition, $\hat\beta_1 - \beta^\star_1 = o_p(1)$, we have that
    \begin{eqnarray*}
        \textup{Term II} = o_p(1).
    \end{eqnarray*}

    For Term III, similarly, we can show that
    \begin{eqnarray*}
        \textup{Term III} = o_p(1).
    \end{eqnarray*}

    Therefore, we have that
    \begin{eqnarray*}
        \frac{n\kappa_n^{-1}\hat{v}^{(d_1)} - n\kappa_n^{-1}v_{\lim}^{(d_1)}}{n\kappa_n^{-1}v_{\lim}^{(d_1)}} = \frac{n\kappa_n^{-1}(\hat{v}^{(d_1)} - v_{\lim}^{(d_1)})}{n\kappa_n^{-1}v_{\lim}^{(d_1)}} = O_p(\frac{o_p(1)}{\Theta(1)}) = o_p(1).
    \end{eqnarray*} 
    Similarly, we can show the counterpart for $\hat{v}^{(d_2)}$. Now using Cauchy-Schwarz inequality, we have that
    \begin{eqnarray*}
        \var(\hat\tau{}^{(d_1, d_2)}_{\beta_1})&=& \var(\hat\mu{}^{(d_1)}_{\beta_1}) + \var(\hat\mu{}^{(d_2)}_{\beta_2}) - 2\cov(\hat\mu{}^{(d_1)}_{\beta_1}, \hat\mu{}^{(d_2)}_{\beta_2}) \\
        &\le& \var(\hat\mu{}^{(d_1)}_{\beta_1}) + \var(\hat\mu{}^{(d_2)}_{\beta_2}) + 2\sqrt{\var(\hat\mu{}^{(d_1)}_{\beta_1}) \var(\hat\mu{}^{(d_2)}_{\beta_2})}\\
        &=& \left\{\sqrt{\var(\hat\mu{}^{(d_1)}_{\beta_1)}} + \sqrt{\var(\hat\mu{}^{(d_2)}_{\beta_2)}}\right\}^{2}\\
        &\le& \left\{\sqrt{v_{\lim}^{(d_1)}} + \sqrt{v_{\lim}^{(d_2)}}\right\}^{2}\triangleq v_{\lim}^{(d_1, d_2)}.
    \end{eqnarray*}
    Therefore, we have that
    \begin{eqnarray*}
        \frac{\hat{v}^{(d_1, d_2)} - v_{\lim}^{(d_1, d_2)}}{v_{\lim}^{(d_1, d_2)}} = o_p(1), \quad v_{\lim}^{(d_1, d_2)} \ge \var(\hat\tau{}^{(d_1, d_2)}_{\beta_1}).
    \end{eqnarray*}

\end{proof}
\section{Hyperparameter Tuning}\label{sec:hyperparameter_tuning}
In this section, we present how we tune GNN-related models in our simulation study and real data. Our data are splitted into training set, validation set, and testing set. We choose a set of hyperparameters that generate the best performance in validation set.

\paragraph{GAT hyperparameter tuning.} We tune the GAT model to maximize outcome prediction $R^2$ on the simulated data. Table~\ref{tab:gat_tuning} summarizes the hyperparameter search space. We find that a medium-depth architecture (GAT layers [64, 32], representation dimension 16, head layers [64, 32, 16]) with no IPM regularization and no dropout achieves the best balance between model capacity and regularization, recovering substantial predictive power from the graph structure alone.

\begin{table}[h]
    \centering
    \caption{GAT hyperparameter search space}
    \label{tab:gat_tuning}
    \begin{tabular}{ll}
        \toprule
        Hyperparameter & Values Searched \\
        \midrule
        GAT hidden layers & [32], [64], [64, 32], [128, 64] \\
        Representation dim & 16, 32 \\
        Head hidden layers & [16], [32, 16], [64, 32], [64, 32, 16] \\
        Attention heads & 4, 8 \\
        IPM regularization $\lambda$ & 0.0, 0.1, 1.0 \\
        Dropout & 0.0, 0.1, 0.2 \\
        Learning rate & 0.0005, 0.001, 0.002 \\
        Weight decay & 0.0001, 0.01, 0.05, 0.1, 0.2 \\
        Epochs & 500, 800, 1000, 2000 \\
        Early stopping patience & 30, 50, 100 \\
        \bottomrule
    \end{tabular}
\end{table}

\paragraph{PNA hyperparameter tuning.} We tune the PNA model to minimize average confidence interval width across contrasts. Table~\ref{tab:pna_tuning} summarizes the hyperparameter search space. We find that simpler models with fewer aggregators and moderate regularization yield the narrowest confidence intervals.

\begin{table}[h]
    \centering
    \caption{PNA hyperparameter search space}
    \label{tab:pna_tuning}
    \begin{tabular}{ll}
        \toprule
        Hyperparameter & Values Searched \\
        \midrule
        GNN hidden dimension and layers & [32], [64], [128], [256], [128, 64] \\
        Representation dimension & 16, 32, 64 \\
        Head hidden dimension and layers & [], [16], [32], [64], [64, 32]\\
        Aggregators & \{mean\}, \{mean, sum\}, \{mean, sum, max\}, \\
                    & \{mean, sum, max, std\} \\
        Scalers & \{identity\}, \{identity, amplification\}, \\
                & \{identity, amplification, attenuation\} \\
        Learning rate & 0.0002, 0.0005, 0.001, 0.005 \\
        IPM regularization $\lambda$ & 0.0, 0.1, 1.0 \\
        Weight decay & 0.0001, 0.01, 0.05, 0.1, 0.2 \\
        Dropout & 0.0, 0.2, 0.3, 0.5 \\
        Epochs & 500, 800, 1000, 2000 \\
        Early stopping patience & 30, 50, 100 \\
        \bottomrule
    \end{tabular}
\end{table}

\section{Additional Details in Numerical Studies}

\subsection{Synthetic data generating process}
\label{sec:dgp_simulation}
\paragraph{Data generating process.} We generate a network $\bG$ with $n=2000$ units, with average degree $\delta_{\text{avg}}$, and maximum degree $\delta_{\text{max}}$.  Unit-level covariates $X_i = (X_{i1}, X_{i2}, X_{i3})^\T$ are generated as functions of degree $H_i = |\cN_i|$:
\begin{align*}
    X_{i1} \sim \cN(0.5 H_i, 1), \quad X_{i2} \sim \cN(0.02 H_i^2, 1), \quad X_{i3} \sim \textup{Poisson}(\indicator(H_i \geq \delta_{\text{avg}})),
\end{align*}

In all of our simulation studies, to add the complexity of unmeasured confounders, we only use $X_{i3}$ as well as an intercept to build model adjustments, while keeping $X_{i1}$ and $X_{i2}$ as blind. 

Potential outcomes incorporate both individual covariates and neighbor spillover:
\begin{align}\label{eq:dgp_spillover}
    Y_i(d) = \beta_0 + \beta_H H_i + \beta_d d + g(X_i) + \frac{1}{|\cN_i \setminus \{i\}|}\sum_{j \in \cN_i \setminus \{i\}} g(X_j) + \varepsilon_i,
\end{align}
where $g(X_i) = 0.6 X_{i1} + 0.15 X_{i2}^2 + 0.4 \tanh(X_{i3})$ is a nonlinear function, and $\varepsilon_i \sim \cN(0,1)$. The neighbor spillover term $\frac{1}{|\cN_i \setminus \{i\}|}\sum_{j \in \cN_i \setminus \{i\}} g(X_j)$ captures the influence of neighbors' covariates on unit $i$'s outcome, which cannot be captured by individual-level covariates alone. 
Treatment $A_i$ is assigned uniformly at random with $P(A_i = 1) = 1/3$. Exposure $D_i \in \{0, 1, 2\}$ is determined by discretizing the proportion of treated neighbors into $k=3$ bins. 

In Section \ref{sec:simulation}, we choose $\delta_{\text{avg}} = 3$ and $\delta_{\text{max}} = 9$. We estimate $\tau^{(d_1, d_2)} = \E[Y(d_1) - Y(d_2)]$ with $d_1 = 0$ and $d_2 = 2$; the true effect is $\tau^\star = -0.594$. 

\subsection{Additional synthetic study}

To assess robustness across network structures, we repeat the simulation with a medium-density network (average degree 5, maximum degree 10). Table~\ref{tab::simu_out2_results} presents results for this configuration, maintaining the same DGP and 3-level exposure ($k=3$). The true treatment effect remains $\tau^\star = -0.594$.

\begin{table}[h]
    \centering
    \caption{Simulation results with medium-density network (avg degree 5, max degree 10). True $\tau^\star = -0.594$.}
    \small
    \begin{tabular}{ll|ccc|cc}
        \toprule
        Description & Estimator & Bias & SD & RMSE & Coverage & Power \\
        \midrule
        Primitive  & $\hat\tau^{\hort}$ & $-0.359 \pm 0.24$ & $1.475 \pm 0.17$ & $1.518 \pm 0.18$ & 0.38 & 0.60 \\
        & $\hat\tau^{\haj}$ & $-0.027 \pm 0.06$ & $0.367 \pm 0.04$ & $0.368 \pm 0.04$ & 0.94 & 0.21 \\
        \midrule
        Linear  
        & $\hat\tau^{\textup{linear-raw}}$ & $-0.084 \pm 0.08$ & $0.459 \pm 0.05$ & $0.467 \pm 0.05$ & 0.98 & 0.16 \\
        & $\hat\tau^{\textup{linear-linear}}$ & $0.057 \pm 0.04$ & $0.238 \pm 0.03$ & $0.245 \pm 0.03$ & 0.98 & 0.21 \\
        & $\hat\tau^{\textup{linear-pred}}$ & $0.074 \pm 0.04$ & $0.241 \pm 0.03$ & $0.252 \pm 0.03$ & 0.97 & 0.22 \\
        \midrule
        GAT
        & $\hat\tau^{\textup{GAT-raw}}$ & $-0.010 \pm 0.06$ & $0.382 \pm 0.04$ & $0.382 \pm 0.04$ & 0.93 & 0.21 \\
        & $\hat\tau^{\textup{GAT-linear}}$ & $0.113 \pm 0.04$ & $0.221 \pm 0.03$ & $0.248 \pm 0.03$ & 0.92 & 0.27 \\
        & $\hat\tau^{\textup{GAT-pred}}$ & $0.116 \pm 0.04$ & $0.221 \pm 0.03$ & $0.250 \pm 0.03$ & 0.92 & 0.25 \\
        \midrule
        PNA 
        & $\hat\tau^{\textup{PNA-raw}}$ & $\mathbf{-0.008 \pm 0.04}$ & $\mathbf{0.241 \pm 0.03}$ & $\mathbf{0.242 \pm 0.03}$ & 0.97 & 0.81 \\
        & $\hat\tau^{\textup{PNA-linear}}$ & $\mathbf{0.041 \pm 0.03}$ & $\mathbf{0.158 \pm 0.02}$ & $\mathbf{0.163 \pm 0.02}$ & 0.96 & 0.73 \\
        & $\hat\tau^{\textup{PNA-pred}}$ & $\mathbf{0.041 \pm 0.03}$ & $\mathbf{0.159 \pm 0.02}$ & $\mathbf{0.164 \pm 0.02}$ & 0.96 & 0.73 \\
        \bottomrule
    \end{tabular}
    \label{tab::simu_out2_results}
\end{table}

PNA methods continue to dominate with the lowest RMSE ($\approx 0.16$--0.24), representing a 33\% improvement over linear GRR (0.245) and 34\% over GAT methods ($\approx$ 0.25). Notably, PNA-pred achieves the highest power (81\%) despite a slightly higher RMSE than PNA-linear/no-harm, consistent with the variance-power trade-off observed in the sparse network setting. The linear-no-cal method performs poorly (RMSE 0.467) in this multi-level exposure setting, demonstrating that separate OLS fitting per exposure level is unreliable when data is split across three groups. Overall, the results confirm that PNA's performance advantage is robust across network densities.

\subsection{Real data study}
\paragraph{Network structure and exposure distribution.}\label{sec:village-network}
The network is constructed as follows: for each village $i$, we connect it to its nearest five villages that are within 8 kilometers. To avoid issues with extremely high-degree nodes that could dominate the interference structure, we cap the maximum degree by removing edges from villages with excessively high connectivity. The resulting network captures potential spillover pathways through which treatment effects may propagate across nearby villages---for example, through local labor markets, trade relationships, or price effects. 

Table~\ref{tab:network_structure} summarizes the network structure and exposure distribution. The network has 2,285 undirected edges with an average degree of 7.0. The degree distribution is similar across exposure groups, suggesting that network position does not strongly predict exposure assignment.

\begin{table}[h]
    \centering
    \caption{Network structure and exposure distribution. Standard deviations of average degree at different exposure level are presented in $(\cdot)$.}
    \label{tab:network_structure}
    \begin{tabular}{cccccc}
        \toprule
        Exposure & $A$ & $H$ & $n$ & Avg Degree & Degree Range \\
        \midrule
        $D=0$ & 0 & 0 & 97 & 6.92 (1.06) & [3, 8] \\
        $D=1$ & 0 & 1 & 228 & 7.10 (1.14) & [2, 8] \\
        $D=2$ & 1 & 0 & 82 & 6.83 (1.34) & [1, 8] \\
        $D=3$ & 1 & 1 & 246 & 7.00 (1.20) & [3, 8] \\
        \midrule
        \multicolumn{3}{c}{Total} & 653 & 7.00 & [1, 8] \\
        \bottomrule
    \end{tabular}
\end{table}

\paragraph{Propensity score distributions.} Figure~\ref{fig:propensity} summarizes the propensity score distributions. The propensity scores $\pi_i^{(d)} = P(D_i = d)$ are estimated via Monte Carlo simulation with 100,000 rounds. Appendix section~\ref{sec:mc_lambda} illustrates how we use Monte Carlo simulation to approximate $\Lambda$. The propensity scores exhibit substantial heterogeneity, reflecting the complex dependence structure induced by interference. Notably, the $H=0$ exposures ($D=0$ and $D=2$) have lower average propensities because having fewer than one-third of neighbors treated is relatively rare given the network structure and treatment rates.

\begin{figure}[h]
    \centering
    \includegraphics[width=\linewidth]{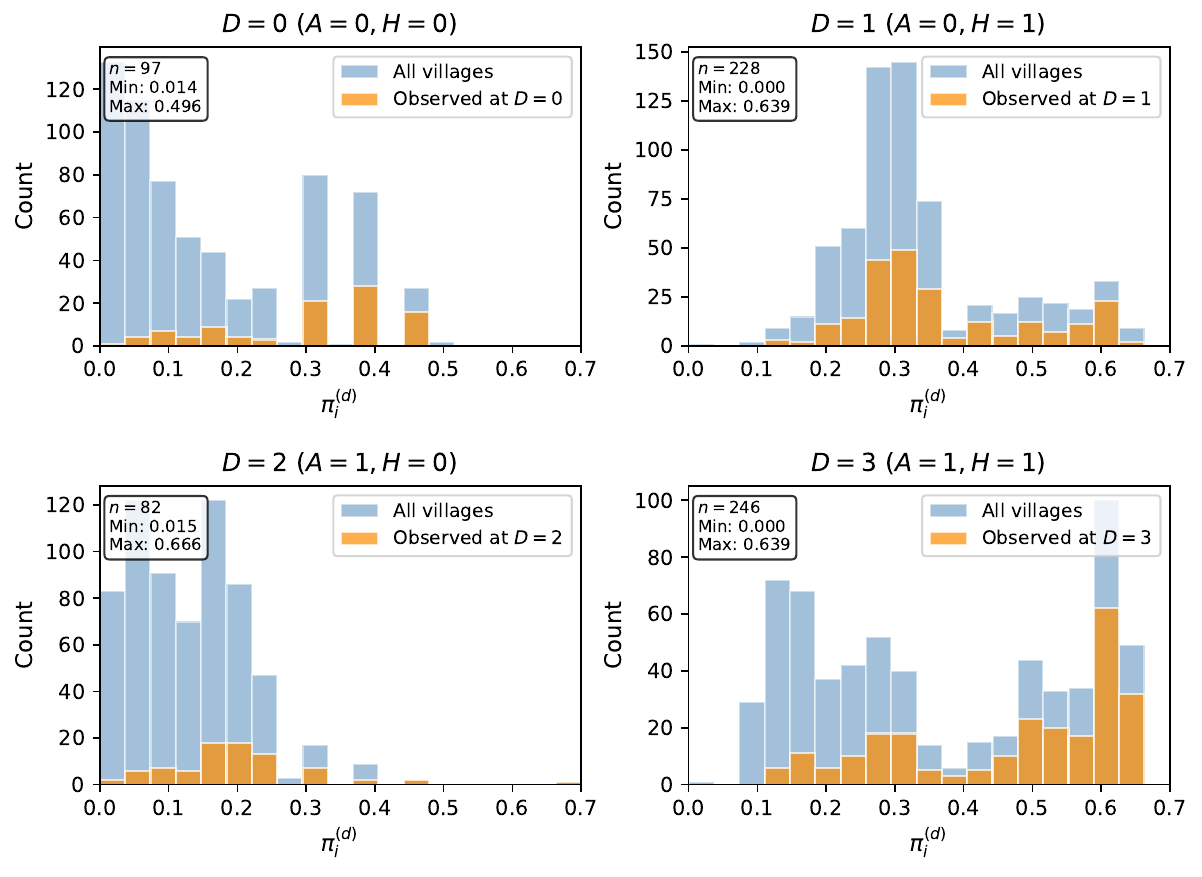}
    \caption{Propensity score distributions by exposure level. Blue bars show the distribution of $\pi_i^{(d)}$ across all 653 villages; orange bars show the distribution among villages actually observed at exposure $D=d$.}
    \label{fig:propensity}
\end{figure}

\paragraph{95\% confidence interval widths.} Table~\ref{tab:revenue_ci_width} presents 95\% confidence interval widths for the treatment effect estimates.
\begin{table}[h!]
    \centering
    \caption{95\% confidence interval widths for the treatment effect estimates.}
    \label{tab:revenue_ci_width}
    \begin{tabular}{lcccc}
        \toprule
        Contrast & HT & H\'{a}jek & Linear & PNA-no-harm \\
        \midrule
        Direct ($H=0$) & 4762 & 3852 & 3926 & \textbf{3638} \\
        Direct ($H=1$) & 2595 & 2573 & 2521 & \textbf{2408} \\
        Spillover ($A=0$) & 3614 & 2836 & 2807 & \textbf{2677} \\
        Spillover ($A=1$) & 3743 & 3589 & 3977 & \textbf{3503} \\
        \midrule
        Average & 3679 & 3213 & 3308 & \textbf{3057} \\
        \bottomrule
    \end{tabular}
\end{table}

\paragraph{A second outcome: wage bill per enterprise.} Table~\ref{tab:wage_results} presents treatment effect estimates for wage bill per enterprise (in KES). For this outcome, the PNA-no-harm estimator requires careful initialization to achieve competitive performance. Using seed 50039 for reproducibility, PNA-no-harm achieves an average confidence interval width of 507 KES, representing a 3\% improvement over H\'{a}jek (521) and 4\% over HT (528). The estimates suggest a moderate positive direct effect when neighbors are untreated ($H=0$), with point estimates around 363 KES. All estimators find insignificant spillover effects, consistent with wage bill being less sensitive to network spillovers than revenue.

\begin{table}[h]
    \centering
    \caption{Treatment effect estimates for wage bill per enterprise (KES). * denotes statistical significance at $\alpha=0.05$. PNA results use seed 50039.}
    \label{tab:wage_results}
    \small
    \begin{tabular}{l cc cc cc cc}
        \toprule
        & \multicolumn{2}{c}{HT} & \multicolumn{2}{c}{H\'{a}jek} & \multicolumn{2}{c}{Linear} & \multicolumn{2}{c}{PNA-no-harm} \\
        \cmidrule(lr){2-3} \cmidrule(lr){4-5} \cmidrule(lr){6-7} \cmidrule(lr){8-9}
        Contrast & $\hat\tau$ & 95\% CI & $\hat\tau$ & 95\% CI & $\hat\tau$ & 95\% CI & $\hat\tau$ & 95\% CI \\
        \midrule
        Direct($H=0$) & 516* & [229, 803] & 436* & [155, 718] & 444* & [158, 730] & 363* & [94, 632] \\
        Direct($H=1$) & -95 & [-337, 146] & 11 & [-229, 250] & 28 & [-218, 274] & 16 & [-204, 237] \\
        Spillover($A=0$) & 606* & [380, 833] & 372* & [150, 594] & 321* & [81, 561] & 253* & [30, 476] \\
        Spillover($A=1$) & -5 & [-307, 297] & -54 & [-353, 245] & -29 & [-378, 321] & -66 & [-367, 235] \\
        \bottomrule
    \end{tabular}
\end{table}

\begin{table}[h]
    \centering
    \caption{95\% confidence interval widths for wage bill (KES). Boldface indicates narrowest width for each contrast.}
    \label{tab:wage_ci_width}
    \begin{tabular}{lcccc}
        \toprule
        Contrast & HT & H\'{a}jek & Linear & PNA-no-harm \\
        \midrule
        Direct($H=0$) & 574 & 563 & 572 & \textbf{538} \\
        Direct($H=1$) & 483 & 479 & 492 & \textbf{442} \\
        Spillover($A=0$) & 453 & \textbf{444} & 480 & 446 (close to best) \\
        Spillover($A=1$) & 603 & \textbf{598} & 700 & 602 (close to best)\\
        \midrule
        Average & 528 & 521 & 561 & \textbf{507} \\
        \bottomrule
    \end{tabular}
\end{table}

\section{Technical appendices and supplementary material}
Technical appendices with additional results, figures, graphs, and proofs may be submitted with the paper submission before the full submission deadline (see above). You can upload a ZIP file for videos or code, but do not upload a separate PDF file for the appendix. There is no page limit for the technical appendices. 

Note: Think of the appendix as ``optional reading'' for reviewers. The paper must be able to stand alone without the appendix; for example, adding critical experiments that support the main claims to an appendix is inappropriate. 



\end{document}